\documentclass[fleqn,reqno]{article}
\usepackage[margin=0.75in]{geometry}
\usepackage{amsmath,amsthm,amssymb,mathrsfs}
\usepackage{enumerate}

\newcommand{\DATUM}{October 31, 2009}              %                        %
\theoremstyle{plain}
\newtheorem{theorem}{Theorem}
\newtheorem{proposition}[theorem]{Proposition}
\newtheorem{lemma}[theorem]{Lemma}

\theoremstyle{definition}

\newtheorem{remark}[theorem]{Remark}

\newcommand{\R}{\mathbb{R}}
\newcommand{\C}{\mathbb{C}}

\newcommand{\DETAILS}[1]{}

\newcommand{\cE}{\mathcal{E}}

\newcommand{\cL}{\mathcal{L}}         %%%%%%%%%%%%%%%%%%%%%%%%
\newcommand{\Lat}{\mathcal{L}}
\renewcommand{\Re}{\operatorname{Re}}
\renewcommand{\Im}{\operatorname{Im}}

\newcommand{\Curl}{\operatorname{curl}}
\newcommand{\Div}{\operatorname{div}}
\newcommand{\Cov}[1]{\nabla_{\!\!#1}}

\newcommand{\Null}{\operatorname{null}}

\newcommand{\Lpsi}[2]{\mathscr{L}_{#2}^{}(\tau)}
\newcommand{\Hpsi}[2]{\mathscr{H}_{#2}^{}(\tau)}
\newcommand{\LA}[2]{\vec{\mathscr{L}}_{}^{}(\tau)}
\newcommand{\HA}[2]{\vec{\mathscr{H}}_{}^{}(\tau)}

\begin{document}
\date{\DATUM}
%%%%%%%%%%%%%%%%%%%%%%%%%%%%%%%%%%%%%%%%%%%%%%%%%%%%%%%%%%%%%%%%%%%%%%%%%%%%%%%%%%%%%%%%%%%%%

\title{Abrikosov Lattice Solutions of the Ginzburg-Landau Equations}
\author{ T. Tzaneteas\footnote{This paper is part of the first author's PhD thesis.}
\thanks{Supported in part by Ontario graduate fellowship and by NSERC under Grant NA 7901 }\,\,\ and\   I. M. Sigal\
\thanks{Supported by NSERC Grant NA7601}
\thanks{Corresponding author; E-mail address: im.sigal@utoronto.ca } \\
\\
Dept. of Mathematics, Univ. of Toronto, Toronto,  Canada, M5S 2E4}
%\author{T. Tzaneteas\footnote{Department of Mathematics, University of Toronto, Toronto, ON, Canada, M5S 2E4}\ \footnote{This paper is part of the first author's PhD thesis.} \ and\  I. M. Sigal\footnote{Department of Mathematics, University of Toronto, Toronto, ON, Canada, M5S 2E4}}
%\date{October, 2009}
%\footnote{This paper is part of the first author's PhD thesis.}
\maketitle

\begin{abstract}
Building on the earlier work of Odeh, Barany, Golubitsky, Turski and Lasher we give a proof of existence of Abrikosov vortex lattices in the Ginzburg-Landau model of superconductivity.

\end{abstract}

%%%%%%%%%%%%%%%%%%%%%%%%%%%%%%%%%%%%%%%%%%%%%%%%%%%%%%%%%%%%%%%%%%%%%%%%%%%%%%%%%%%%%%%%%%%%%

\section{Introduction}
\label{sec:intro}

%%%%%%%%

%\subsection{The Ginzburg-Landau Model}

\textbf{1.1 The Ginzburg-Landau Model.} The Ginzburg-Landau model of superconductivity describes a superconductor contained in $\Omega \subset \R^n$, $n = 2$ or $3$, in terms of
a complex order parameter $\psi : \Omega \to \C$, and a magnetic potential $A : \Omega \to \R^n$\footnote{The Ginzburg-Landau theory is reviewed in every book on superconductivity. For reviews of rigorous results see the papers \cite{CHO, DGP} and the books \cite{SS, FS, JT, Rub}}. %books \cite{SS, H}}.
The key physical quantities for the model are
\begin{itemize}
    \item the density of superconducting pairs of electrons, $n_s := |\psi|^2$;
    \item the magnetic field, $B := \Curl A$;
    \item and the current density, $J := \Im(\bar{\psi}\Cov{A}\psi)$.
\end{itemize}
In the case $n = 2$, $\Curl A := \frac{\partial A_2}{\partial x_1} - \frac{\partial A_1}{\partial x_2}$ is a scalar-valued function.
The covariant derivative $\Cov{A}$ is defined to be $\nabla - iA$.
The Ginzburg-Landau theory specifies that a state $(\psi, A)$, in the absence of an external magnetic field, has energy
\begin{equation}
\label{eq:GL-energy}
    \mathcal{E}_\Omega(\psi, A) := \int_\Omega |\Cov{A}\psi|^2 + |\Curl A|^2 + \frac{\kappa^2}{2} (1 - |\psi|^2)^2,
\end{equation}
where $\kappa$ is a positive constant that depends on the material properties of the superconductor.

It follows from the Sobolev inequalities that for bounded open sets $\Omega$, $\mathcal{E}_\Omega$ is well-defined and $C^\infty$ as a
functional on the Sobolev space $H^1$. The critical points of this functional must satisfy the well-known Ginzburg-Landau equations
inside $\Omega$:
\begin{subequations}
\label{eq:GL-equations}
    \begin{equation}
    \label{eq:GLpsi}
      \Delta_A \psi = \kappa^2(|\psi|^2-1)\psi,
    \end{equation}
    \begin{equation}
    \label{eq:GLA}
        \Curl^*\Curl A = \Im(\bar{\psi}\Cov{A}\psi).
    \end{equation}
\end{subequations}
Here $\Delta_A=- \Cov{A}^*\Cov{A},\ \Cov{A}^*$ and $\Curl^*$ are the adjoints of $\Cov{A}$ and $\Curl$. Explicitly, $\Cov{A}^*F = -\Div F + iA\cdot F$, and
$\Curl^* F = \Curl F$ for $n = 3$ and $\Curl^* f = (\frac{\partial f}{\partial x_2}, -\frac{\partial f}{\partial x_1})$ for $n = 2$.

There are two immediate solutions to the Ginzburg-Landau equations that are homogeneous in $\psi$. These are the perfect superconductor
solution where $\psi_S \equiv 1$ and $A_S \equiv 0$, and the \textit{normal} (or non-superconducting) solution where $\psi_N = 0$ and $A_N$ is such that
$\Curl A_N =: b$ is constant.

It is well-known that there exists a
critical value $\kappa_c$ (in the units used here, $\kappa_c = 1/\sqrt{2}$), that separates superconductors into two classes with
different properties: Type I superconductors, which have $\kappa < \kappa_c$ and exhibit first-order phase transitions from the
non-superconducting state to the superconducting state, and Type II superconductors, which have $\kappa > \kappa_c$ and exhibit
second-order phase transitions and the formation of vortex lattices. Existence of the vortex lattice solutions is the subject of the present paper.

%%%%%%%%

%\subsection{Results}

%In this paper we prove the existence of Abrikosov lattice solutions of the Ginzburg-Landau equations in the case of single flux quantum per lattice cell. These solutions were discovered in 1957 by A. Abrikosov, who concluded,
\textbf{1.2 Results.} In 1957, Abrikosov \cite{Abr} discovered solutions of \eqref{eq:GL-equations} whose physical characteristics $n_s$, $B$, and $J$ are periodic with respect to a two-dimensional lattice, while independent of the third dimension, and which have a single flux per lattice cell\footnote{Such solutions correspond cylindrical samples. In 2003, Abrikosov received the Nobel Prize for this discovery}. (In what follows we call such solutions, with $n_s$ and $B$ non-constant, \textit{lattice solutions}, or, if a lattice $\Lat$ is fixed, $\Lat$-\textit{lattice solutions}. In physics literature they are called variously mixed states, Abrikosov mixed states,  Abrikosov vortex
states.) Due to an error of calculation he concluded that the lattice which gives the minimum average energy per lattice cell is the square lattice. Abrikosov's error was corrected by Kleiner, Roth, and Autler \cite{KRA}, who showed that it is in fact the triangular lattice which minimizes the energy.

Since then these Abrikosov lattice solutions have been studied in numerous experimental and theoretical works. Of more mathematical studies, we mention the articles of Eilenberger \cite{Eil} and Lasher \cite{Lash}.

The rigorous investigation of Abrikosov solutions began soon after their discovery. Odeh \cite{Odeh} proved the existence of non-trivial minimizers and obtained a result concerning the bifurcation of solutions at the critical field strength. Barany, Golubitsky, and Tursky \cite{BGT} investigated this bifurcation for certain lattices using equivariant bifurcation theory, and Tak\'{a}\u{c} \cite{Takac} has adapted these results to study the zeros of the bifurcating solutions. %The above results were further improved in \cite{Dutour}.

Except for a variational result of \cite{Odeh} (see also \cite{Dutour2}), work done by both physicists and mathematicians has followed the general strategy of \cite{Abr}.

In this paper we combine and extend the previous technique to give a self-contained proof of the existence of Abrikosov lattice solutions. %and improved version of the resulted presented in \cite{Odeh}. %We also establish existence of Abrikosov lattice with two magnetic flux quanta per cell. The latter result can be extended to arbitrary number of quanta per cell.
To formulate our results we mention that %have to introduce some notation.
lattices $\Lat \subset \R^2$ are characterized by the area $|\Omega_\Lat|$ of the lattice cell $\Omega_\Lat$ and the shape $\tau$, given by the ratio of basis vectors identified as complex numbers (for details see Section \ref{sec:lattice states}). We will prove the following results, whose precise formulation will be given below (Theorem \ref{thm:result}).

%\begin{theorem}
%\label{thm:existence}
%	Suppose that $0 < bn^2 < \kappa^2$. Then for every $\tau$, there exists a mixed-state %Abrikosov lattice solution of type $(\tau, b, n)$ that minimizes the energy per cell.
%\end{theorem}

%\begin{theorem}
%\label{thm:branch}
%Given a lattice, assume there is one flux quantum per the lattice cell. If the magnetic flux per unit area,  $b$, is less than but sufficiently close to $\kappa^2$, then the only non-normal solutions close to the normal state form a branch, parametrized by a complex number, one branch for each lattice.
%\end{theorem}
%
%\begin{theorem}
%\label{thm:minimizing-shape}
%	Of the solutions in the previous theorem, %let $\tau(b)$ be the shape that minimizes the energy at a fixed $b$. Then as $b \to \kappa^2$, $\tau(b)$
%the least energy ones have lattices approaching the triangular (or hexagonal) lattice as $b \to \kappa^2$.
%\end{theorem}

\begin{theorem}
\label{thm:main-result}
	Let $\Lat$ be a lattice with $\left| |\Omega_\Lat| - \frac{2\pi}{\kappa^2} \right| \ll 1$.
	\begin{enumerate}[(I)]
	\item If $|\Omega_\Lat| > \frac{2\pi}{\kappa^2}$, there exists an $\Lat$-lattice solution. If $|\Omega_\Lat| \leq \frac{2\pi}{\kappa^2}$, then there is no $\Lat$-lattice solution in a neighbourhood of the branch of normal solutions.
	\item The above solution is close to the branch of normal solutions and is unique, up to symmetry, in a neighbourhood of this branch.
	\item The solutions above are real analytic in $|\Omega_\Lat|$ in a neighbourhood of $\frac{2\pi}{\kappa^2}$.
	% and $\tau$.
	\item The lattice shape for which the average energy per lattice cell is minimized approaches the triangular lattice as $|\Omega_\Lat| \to \frac{2\pi}{\kappa^2}$.
	\end{enumerate}	
\end{theorem}

\begin{remark}
	\hfill
	\begin{enumerate}[(a)]
	\item \cite{Odeh, Dutour} showed that for all $|\Omega_\Lat| > \frac{2\pi}{\kappa^2}$ there exists a global minimizer of $\mathcal{E}_{\Omega_\Lat}$.
	\item \cite{Odeh, BGT} proved results related to our solutions in (I).
	\item \cite{Lash} proved partial results on (IV).
	\end{enumerate}
\end{remark}

Among related results, a relation of the Ginzburg-Landau minimization problem, for a fixed, finite domain and for increasing Ginzburg-Landau parameter $\kappa$ and external magnetic field, to the Abrikosov lattice variational problem was obtained in \cite{AS, Al2}.   \cite{Dutour}  (see also \cite{Dutour2}) have found boundaries between superconducting, normal and mixed phases.

 All the rigorous results above deal with Abrikosov lattices with one quantum of magnetic flux per lattice cell. partial results for higher magnetic fluxes were proven in \cite{Ch, Al}. This problem will be addressed in our subsequent paper.
%%%%%%%%%%%%%%%%%%%%%%%%%%%%%%%%%%%%%%%%%%%%%%%%%%%%%%%%%%%%%%%%%%%%%%%%%%%%%%%%%%%%%%%%%%%%%

%%%%%%%%%%%%%%%%%%%%%%%%%%%%%%%%%%%%%%%%%%%%%%%%%%%%%%%%%%%%%%
\noindent \textbf{Acknowledgements} \\
The second author is grateful to  Yuri Ovchinnikov for many fruitful discussions. A part of this work was done during I.M.S.'s stay at the IAS, Princeton. %Research of I.M.S. is supported by NSERC under Grant NA 7901. T.T. is supported in part by Ontario graduate fellowship and by NSERC under Grant NA 7901.

%%%%%%%%%%%%%%%%%%%%%%%%%%%%%%%%%%%%%%%%%%%%%%%%%%%%%%%%%%%%%%%%%%%%%%%%%%%%%%%%%%%%%%%%%%%%%

\section{Properties of the  Ginzburg-Landau Equations}
\label{sec:problem}

%%%%

%%%%

%\subsection{Symmetries} %of the Ginzburg-Landau Equations

\textbf{2.1 Symmetries.} The Ginzburg-Landau equations exhibit a number of symmetries, that is, transformations which map solutions to solutions. The most important of these
symmetries is the gauge symmetry, defined for any sufficiently regular function $\eta : \Omega \to \R$, which maps
$(\psi, A) \mapsto (T_\eta\psi, T_\eta A)$, where
\begin{equation}
\label{eq:gauge-symmetry}
    T_\eta\psi = e^{i\eta}\psi, \qquad T_\eta A = A + \nabla\eta.
\end{equation}
There are also the translation symmetry, defined for each $t \in \R^2$, which maps $(\psi, A) \mapsto (T_t\psi, T_t A)$, where
\begin{equation}
\label{eq:translation-symmetry}
    T_t\psi(x) := \psi(x + t), \qquad\qquad T_t A(x) := A(x + t),
\end{equation}
and rotation and reflection symmetry, defined for each $R \in O(2)$ (the set of orthogonal $2 \times 2$ matrices), which maps
$(\psi, A) \mapsto (T_R\psi, T_R A)$, where
\begin{equation}
\label{eq:rotation-reflection-symmetry}
    T_R\psi(x) := \psi(Rx), \qquad\qquad T_R A(x) := R^{-1}A(Rx).
\end{equation}

%%%%

%\subsection{Flux quantization}

\textbf{2.2 Flux Quantization.} One can show that under certain boundary conditions (e.g., 'gauge-periodic', see below, or if $\Omega = \R^2$ and $\mathcal{E}_\Omega < \infty$) %or $\Omega$ is a lattice cell for a lattice state)
the magnetic flux through $\Omega$ is quantized.

%%%%%%%%%%%%%%%%%%%%%%%%%%%%%%%%%%%%%%%%%%%%%%%%%%%%%%%%%%%%%%%%%%%%%%%%%%%%%%%%%%%%%%%%%%%%%

\section{Lattice States}
\label{sec:lattice states}

Our focus in this paper is on states $(\psi, A)$ defined on all of $\R^2$, but whose physical properties, the density of superconducting pairs of electrons, $n_s := |\psi|^2$, the magnetic field, $B := \Curl A$, and the current density, $J := \Im(\bar{\psi}\Cov{A}\psi)$, are doubly-periodic with respect
to some lattice $\cL$. We call such states $\cL-$\emph{lattice states}.
%We now turn to studying quasiperiodic states and we begin by giving a more mathematically convenient definition of a %quasiperiodic state.

%\begin{definition}
One can show that  a state $(\psi, A) \in H^1_{loc}(\R^2;\C) \times H^1_{loc}(\R^2;\R^2)$ is a $\mathcal{L}$-lattice state if and only if translation by an element of the lattice results in a gauge transformation
    of the state, that is, for each $t \in \mathcal{L}$, there exists a function $g_t \in H^2_{loc}(\R^2;\R)$
    such that $$\psi(x + t) = e^{ig_t(x)}\psi(x)\ \mbox{and}\ A(x+t) = A(x) + \nabla g_t(x)$$ almost everywhere.
%\end{definition}
%
%
\DETAILS{With reference to the definition given in section \ref{sec:problem}, we note that any state that is a lattice state with respect to
this definition will have doubly periodic physical properties, and, as a partial converse, given a state whose physical
properties are doubly periodic, then there is a state that is a lattice state according to this definition and that shares those physical properties.}

It is clear that the gauge, translation, and rotation symmetries of the Ginzburg-Landau equations map lattice states to
lattice states. In the case of the gauge and translation symmetries, the lattice with respect to which the solution is
periodic does not change, whereas with the rotation symmetry, the lattice is rotated as well. It is a simple calculation
to verify that the magnetic flux per cell of solutions is also preserved under the action of these symmetries.

%%%%
Note that $(\psi, A)$ is defined by its restriction to a single cell and can be reconstructed from this restriction by lattice translations.

%\subsection{Flux quantization}

%One shows in a standard way
\textbf{3.1 Flux quantization.}  The important property of lattice states is that the magnetic flux through a lattice cell is quantized:
%Here the magnetic flux per cell $\Phi(A)$ is defined to be
\begin{equation}
\label{eq:flux-per-cell}
    \Phi(A) := \int_\Omega \Curl A = 2\pi n
\end{equation}
for some integer $n$. Here $\Omega$ is any fundamental cell of the lattice.
%We will see (in section \ref{sec:quasiperiodicity}) that the magnetic flux of quasiperiodic states is quantized,
%We prove this in the following proposition.
%\begin{proposition}
%\label{prop:flux-quantization}
%    Let $(\psi, A)$ be a $\mathcal{L}$-quasiperiodic state such that $|\psi| > 0$ on $\partial\Omega$. Then for some integer $n$,
%    \begin{align}
%    \label{QS:quant}
%        \int_\Omega \Curl A \,dx = 2\pi n.
%    \end{align}
%\end{proposition}
%\begin{proof}
%    We can define $\theta(x)$ by $\psi(x) = |\psi(x)|e^{i\theta(x)}$, $0 \leq \theta(x) < 2\pi$, whenever $\psi(x) \neq 0$. A simple calculation then
%    gives $|\psi|^2\nabla\theta = \Im\{ \bar{\psi}\nabla\psi \}$ and therefore
%    $J = \Im\{ \bar{\psi}\nabla\psi \} - |\psi|^2A = |\psi|^2(\nabla\theta - A)$.
%    Since $J$ and $|\psi|$ both satisfy periodic boundary conditions on $\partial\Omega$, it follows that $\nabla\theta - A$ does as well, so
%    \begin{align*}
%        \int_{\partial\Omega} \nabla\theta \cdot d\sigma = \int_{\partial\Omega} A \cdot d\sigma
%            = \int_\Omega \Curl A \,dx,
%    \end{align*}
%    where the last step is just Green's Theorem. But since $\psi$ is single-valued, it follows that
%    $\int_{\partial\Omega} \nabla\theta \cdot d\sigma = 2\pi n$ for some integer $n$ and that proves the theorem.
%\end{proof}
Indeed, if $|\psi| > 0$ on the boundary of the cell, we can write %define the phase $\theta$ by the conditions
$\psi = |\psi|e^{i\theta}$ and $0 \leq \theta < 2\pi$. The periodicity of $n_s$ and $J$ ensure the periodicity of $\nabla\theta - A$ and therefore by Green's theorem, $\int_\Omega \Curl A = \oint_{\partial\Omega} A = \oint_{\partial\Omega} \nabla\theta$ and this function is equal to $2\pi n$ since $\psi$ is single-valued.

We let $b$ be the average magnetic flux per lattice cell, $b := \frac{1}{|\Omega|} \Phi(A)$. Equation \eqref{eq:flux-per-cell} then imposes a condition on the area of a cell, namely,
\begin{equation*}
	|\Omega| = \frac{2\pi n}{b}.
\end{equation*}
Due to the physical interpretation of $b$ as being related to the applied magnetic field, from now on we use $b = \frac{2\pi n}{|\Omega|}$ as a parameter of our problem. We note that due to the reflection symmetry of the problem we can assume that $b \geq 0$.

%%%%

%\subsection{Lattice Shape} \label{sec:lattice-shape}

\textbf{3.2 Lattice Shape.} In order to define the shape of a lattice, we identify $x \in \R^2$ with $z = x_1 + ix_2 \in \C$, and view $\mathcal{L}$ as a subset of $\C$. It is a well-known fact (see \cite{Ahlfors}) any lattice $\mathcal{L} \subseteq \C$ can be given a basis ${r, r'}$ such
that the ratio $\tau = \frac{r'}{r}$ satisfies the inequalities:
\begin{enumerate}[(i)]
\item $|\tau| \geq 1$.
\item $\Im\tau > 0$.
\item $-\frac{1}{2} < \Re\tau \leq \frac{1}{2}$, and $\Re\tau \geq 0$ if $|\tau| = 1$.
\end{enumerate}
Although the basis is not unique, the value of $\tau$ is, and we will use that as a measure of the shape of the lattice.

Using the rotation symmetry we can assume that if $\mathcal{L}$ has has as a basis $\{\, re_1, r\tau \,\}$, where $r$ is a positive real number and $e_1 = (1,0)$.

%%%%

%\subsection{Fixing the Gauge}

\textbf{3.3 Fixing the Gauge.}  The gauge symmetry allows one to fix solutions to be of a desired form. We will use the following preposition, first used by \cite{Odeh} and proved in \cite{Takac}. We provide an alternate proof in Appendix \ref{sec:alternate-proof}.

\begin{proposition}
    \label{thm:fix-gauge}
    Let $(\psi, A)$ be an $\mathcal{L}$-lattice state, and let $b$ be the average magnetic flux per cell. Then there is
    a $\mathcal{L}$-lattice state $(\phi, A_0 + a)$ that is gauge-equivalent to a translation of $(\psi, A)$, where
    $A_0(x) = \frac{b}{2}x^\perp$ (where $x^\perp = (-x_2, x_1)$), and $\phi$ and $a$ satisfy the following conditions.
    \begin{enumerate}
    \item[(i)] $a$ is doubly periodic with respect to $\mathcal{L}$: $a(x + t) = a(x)$ for all $t \in \mathcal{L}$.
    \item[(ii)] $a$ has mean zero: $\int_\Omega a = 0$.
    \item[(iii)] $a$ is divergence-free: $\Div a = 0$.
    \item[(iv)] $\phi(x + t) = e^{\frac{ib}{2}t\wedge x}\phi(x)$, where $t\wedge x = t_1x_2 - t_2x_1$, for $t = re_1, r\tau$.
    \end{enumerate}
\end{proposition}

%%%%

%\subsection{Lattice energy}

\textbf{3.3 Lattice Energy.} Lattice states clearly have infinite total energy, so we will instead
consider the average energy per cell,  defined by
\begin{equation}
\label{eq:GL-energy-per-cell}
    E(\psi, A) := \frac{1}{|\Omega|} \mathcal{E}_\Omega(\psi, A).
\end{equation}
Here, $\Omega$ is a primitive cell of the lattice with respect to which $(\psi, A)$ is a lattice state and $|\Omega|$ is its Lebesgue measure.
We seek minimizers of this functional under the condition that the average magnetic flux per lattice cell is fixed:
$\frac{1}{|\Omega|}\Phi(A) = b$.

In terms of the minimization problem, we see that the perfect superconductor is a solution only when $\Phi(A)$ is
fixed to be $\Phi(A)= 0$. On the other hand, there is a normal solution,  $(\psi_N = 0,\ A_N,\ \Curl A_N =$ constant), for any condition on $\Phi(A)$.

We define the energy of the lattice with the flux $n$ per cell as
\begin{equation}
\label{eq:GL-energy-of-lattice}
    \cE_n (\cL ) := \inf E(\psi, A),
%        \cE_n (\cL ) := \inf_{(\psi, a)\in \mathscr{H}_n(\mathcal{L})} E(\psi, A),
\end{equation}
where the infimum is taken over all smooth $\Lat$-lattice states satisfying
%. $\mathscr{H}_n(\mathcal{L})$ is the closure of the space of smooth pairs $(\psi, A)$ defined on $\Omega_\Lat$, any fundamental lattice cell, and satisfying
(i) through (iv) of Proposition \ref{thm:fix-gauge}.

%\subsection{Result}

\textbf{3.4 Result. Precise Formulation.} The following theorem gives the precise formulation of
%Theorems \ref{thm:branch}  and \ref{thm:minimizing-shape} on Introduction.
Theorem \ref{thm:main-result} from the introduction.
\begin{theorem}
\label{thm:result}
    Let $n = 1$.
    \begin{enumerate}[(I)]
	\item For every $b$ sufficiently close to but less than the critical value $b_c = \kappa^2$, there exists an $\cL-$lattice solution of the Ginzburg-Landau equations with one quantum of flux per cell and with average magnetic flux per cell equal to $b$.
	\item This solution is unique, up to the symmetries, in a neighbourhood of the normal solution.
	\item The family of these solutions is real analytic in $b$ in a neighbourhood of $b_c$.
	\item If $\kappa^2 > 1/2$, then the global minimizer $\Lat_b$ of the average energy per cell, $\cE_1(\mathcal{L})$, approaches the $\Lat_{triangular}$ as $b \to b_c$ in the sense that the shape $\tau_b $ approaches $ \tau_{triangular} = e^{i\pi/3}$ in $\C$.
	\end{enumerate}
%    Moreover, if $\kappa > \kappa_c$, then for any fixed $b$, the global minimizer of the average energy per cell is the solution
%    corresponding to the triangular lattice, while the solution corresponding to the rectangular lattice is a saddle point of the energy functional $\cE_1 (\cL )$.
\end{theorem}
The rest of this papers is devoted to the proof of this theorem.

%\textbf{Comment. How are solutions for various lattices related to critical points to the GLE? Unless they are critical points of $\cE_n (\cL )$ variation in $\cL$ would change the energy in the 1st order??}

%%%%%%%%%%%%%%%%%%%%%%%%%%%%%%%%%%%%%%%%%%%%%%%%%%%%%%%%%%%%%%%%%%%%%%%%%%%%%%%%%%%%%%%%%%%%%

\section{Rescaling}
\label{sec:reduced-problem}

In this section we rescale $(\psi, A)$ to eliminate the dependence of the size of the lattice on $b$. Our space will then depend only on the number of quanta of flux and the shape of the lattice.

Suppose, then, that we have a $\mathcal{L}$-lattice state $(\psi, A)$, where $\mathcal{L}$ has shape $\tau$. Now let $b$ be the average magnetic flux per cell of the state and $n$ the quanta of flux per cell. From the quantization of the flux, we know that
\begin{equation*}
    b = \frac{2\pi n}{|\Omega|} = \frac{2\pi n}{r^2 \Im\tau},
\end{equation*}
We set
$\sigma := \left(\frac{n}{b}\right)^{\frac{1}{2}}$. The last two relations give $\sigma = \left(\frac{\Im\tau}{2\pi}\right)^{\frac{1}{2}} r$. We now define the rescaling $(\hat{\psi}, \hat{A})$ to be
\begin{equation*}
    (\hat{\psi}(x), \hat{A}(x)) := ( \sigma \psi(\sigma x), \sigma A(\sigma x) ).
\end{equation*}
Let $\mathcal{L}^\tau$ be the lattice spanned by $r^\tau$ and $r^\tau\tau$, with $\Omega^\tau$ being a primitive cell of that lattice. Here %$r^\tau$ to be
    \begin{equation}
    \label{eq:rtau}
        r^\tau := \left( \frac{2\pi}{\Im\tau} \right)^{\frac{1}{2}}.
    \end{equation}
    We note that $|\Omega^\tau| = 2\pi n$.
    We summarize the effects of the rescaling above: %in the following proposition whose proof is straightforward and therefore omitted.
%\begin{proposition}
%\label{prop:rescale}

    %Then:
    \begin{enumerate}[(i)]

    \item $(\hat{\psi}, \hat{A})$ is a $\mathcal{L}^\tau$-lattice state.

    \item   $E(\psi,A) = \hat{\mathcal{E}}_{\lambda}(\hat{\psi},\hat{A})$, where $\lambda = \frac{\kappa^2 n}{b}$ and
        \begin{equation}
        \label{eq:G-energy}
            \hat{\mathcal{E}}_\lambda(\psi, A) = \frac{\kappa^4}{2\pi \lambda^2} \int_{\Omega^\tau}\left( |\Cov{A}\psi|^2 + |\Curl A|^2
                        + \frac{\kappa^2}{2} ( |\psi|^2 - \frac{\lambda}{\kappa^2} )^2\right) \,dx.
        \end{equation}

    \item $\psi$ and $A$ solve the Ginzburg-Landau equations if and only if $\hat{\psi}$ and $\hat{A}$ solve
            \begin{subequations}
            \label{eq:rGL}
            \begin{equation}
            \label{rGLpsi}
                (-\Delta_{A}  - \lambda) \psi = -\kappa^2 |\psi|^2\psi,
            \end{equation}
            \begin{equation}
            \label{rGLA}
                \Curl^*\Curl A = \Im\{ \bar{\psi}\Cov{A}\psi \}
            \end{equation}
        \end{subequations}
 for  $\lambda = \frac{\kappa^2 n}{b}$. The latter equations are valid on  $\Omega^\tau$  with the boundary conditions given in the next statement.

    \item   \label{reduced-gauge-form}
        If $(\psi, A)$ is of the form described in Proposition \ref{thm:fix-gauge}, then %$(\hat{\psi}, \hat{A})$ is such that
        $$\hat{A} = A^n_0 + a,\        \mbox{where}\ A^n_0(x) := \frac{n}{2} x^\perp,$$  where $x^\perp =(-x_2, x_1)$, and $\hat\psi$ and $a$ satisfy
            \begin{enumerate}[(a)]
            \item $a$ is double periodic with respect to $\mathcal{L}^\tau$,
            \item $\int_{\Omega^\tau} a = 0$,
            \item $\Div a = 0$,
            \item \label{reduced-quasiperiodic-bc}
                $\hat\psi(x + t) = e^{\frac{in}{2}t\wedge x}\hat\psi(x)$ for $t = r^\tau$, $r^\tau\tau$.
            \end{enumerate}

    %\item Finally the average magnetic flux per cell of $(\hat{\psi},\hat{A})$ is
%       \begin{equation}
%       \label{RP:Flux}
%           \frac{1}{|\Omega^\tau|} \int_{\Omega^\tau} \Curl \hat{A} \,dx = n.
%       \end{equation}
    \end{enumerate}
%\end{proposition}

%We also note that if $(\psi, A^n_0 + P)$ is in the fixed gauge, then
%\begin{equation}
%\label{RP:curlPmean}
%   \int_{\Omega^\tau} \Curl P \,dx = \int_{\Omega^\tau} \Curl A \,dx - \int_{\Omega^\tau} A^n_0 \,dx = 2\pi n - 2\pi n = 0,
%\end{equation}
%as we will use this fact below.

%\subsection{Statement of problem}
\emph{In what follows we drop the hat from $\hat{\psi}$, $\hat{A}$, and $\hat{\mathcal{E}}_\lambda$.}

We are now state our problem in terms of the fields $\psi$ and $a$. We define the Hilbert space $\Lpsi{2}{n}$ to be the closure under the $L^2$-norm of the space of all smooth
$\psi$ on $\Omega^\tau$ satisfying the quasiperiodic boundary condition (d) in part \eqref{reduced-gauge-form} above. %of Proposition \ref{prop:rescale}.
$\Hpsi{2}{n}$ is then the space of all $\psi \in \Lpsi{2}{n}$ whose (weak) partial derivatives up to order $2$ are square-integrable.

Similarly, we define the Hilbert space $\LA{p}{\Div,0}$ to be the closure of the space of all smooth $a$ on $\Omega^\tau$ that satisfy
periodic boundary conditions, have mean zero, and are divergence free, and $\HA{2}{\Div,0}$ is then the subspace of $\LA{2}{\Div,0}$ consisting of those elements whose partial derivatives up to order $2$ are square-integrable.

Our problem then is, for each $n = 1,2,\ldots$,  find $(\psi, a) \in \Hpsi{2}{n} \times \HA{2}{\Div,0}$ so that $(\psi, A^n_0 + a)$
solves the rescaled Ginzburg-Landau equations \eqref{eq:rGL}, and among these find the one that minimizes the average energy $\mathcal{E}_\lambda$.

%%%%%%%%%%%%%%%%%%%%%%%%%%%%%%%%%%%%%%%%%%%%%%%%%%%%%%%%%%%%%%%%%%%%%%%%%%%%%%%%%%%%%%%%%%%%%

\section{Reduction to Finite-dimensional Problem}
\label{sec:reduction}

In this section we reduce the problem of solving Eqns \eqref{eq:rGL} to a finite dimensional problem. We address the latter in the next section. Substituting $A = A^n_0 + a$, we rewrite  \eqref{eq:rGL} as
\DETAILS{We begin by introducing the
maps $F_1 : \R \times \Hpsi{2}{n} \times \HA{2}{\Div,0} \to \Lpsi{2}{n}$
and $F_2 : \HA{2}{\Div,0} \times \Hpsi{2}{n} \to \LA{2}{\Div,0}$, which are defined to be}
\begin{subequations}
    \begin{equation}
    \label{eq:F-1}
        %F_1(\lambda, \psi, a) :=
        (L^n - \lambda)\psi + 2ia\cdot\nabla_{A^n_0}\psi %- 2A^n_0\cdot P\psi
        + |a|^2\psi + \kappa^2|\psi|^2\psi =0,
    \end{equation}
    \begin{equation}
    \label{eq:F-2}
        %F_2(\psi, a) :=
        (M + |\psi|^2)a - \Im\{ \bar{\psi}\Cov{A^n_0}\psi \}=0,
    \end{equation}
\end{subequations}
where
\begin{equation}
\label{OP:LN}
    L^n := -\Delta_{A^n_0}
%\end{equation}
\mbox{ and }
%\begin{equation}
%\label{OP:M}
    M := \Curl^*\Curl.
\end{equation}
\DETAILS{Equations \eqref{eq:rGL} is then equivalent to
\begin{subequations}
\label{eq:F-i=0}
    \begin{equation}
    \label{eq:F-1=0}
        F_1(\lambda, \psi, a) = 0,
    \end{equation}
    \begin{equation}
    \label{eq:F-2=0}
        F_2(\psi, a) = 0.
    \end{equation}
\end{subequations}

We summarize the properties of these two maps in the following proposition, whose straightforward proof is omitted.
\begin{proposition}
    \hfill
    \begin{enumerate}[(a)]
    \item $F_1$ and $F_2$ are $C^\infty$,
    \item for all $\lambda$, $F_1(\lambda, 0, 0) = 0$ and $F_2(0,0) = 0$,
    \item for all $\alpha \in \R$, $F_1(\lambda, e^{i\alpha}\psi, a) = e^{i\alpha}F_1(\lambda, \psi, a)$ and
            $F_2(e^{i\alpha}\psi, a) = F_2(\psi, a)$.
    \end{enumerate}
\end{proposition}}
The operators $L^n$ and $M$ are elementary and well studied. Their properties that will be used below are summarized in the following theorems, whose proofs may be found in Appendix \ref{sec:operators}.

\begin{theorem}
\label{thm:L-theorem}
    $L^n$ is a self-adjoint operator on $\Hpsi{2}{n}$ with spectrum $\sigma(L^n) = \{\, (2k + 1)n : k = 0, 1, 2, \ldots \,\}$ and
    $\dim_\C \Null (L^n - n) = n$.
\end{theorem}

\begin{theorem}
\label{thm:M-theorem}
    $M$ is a strictly positive operator on $\HA{2}{\Div,0}$ with discrete spectrum.
\end{theorem}

%%%%

%\subsection{Solving $F_2 = 0$}

We first solve the second equation \eqref{eq:F-2} for $a$ in terms of $\psi$, using the fact that $M$ is a strictly positive operator, and
therefore $M + |\psi|^2$ is invertible. We have $a=a(\psi)$, where
\begin{equation}
\label{eq:a=}
    a(\psi) = (M + |\psi|^2)^{-1}\Im(\bar{\psi}\Cov{A^n_0}\psi).
\end{equation}
We collect the elementary properties of the map $a$ in the following preposition, where we identify  with a real Banach space using $\psi \leftrightarrow \overrightarrow{\psi}:=(\Re \psi, \Im \psi)$.
\begin{proposition}
 The unique solution, $a(\psi)$, of \eqref{eq:F-2} maps $\Hpsi{2}{n}$ to $\HA{2}{\Div,0}$ and
 has the following properties:
    \begin{enumerate}[(a)]
    \item $a(\cdot)$ is analytic as a map between real Banach spaces.
    \item $a(0) = 0$.
    \item For any $\alpha \in \R$, $a(e^{i\alpha}\psi) = a(\psi)$.
    \end{enumerate}
\end{proposition}
\begin{proof}
    The only statement that does not follow immediately from the definition of $a$ is (a). It is clear that $\Im(\bar{\psi}\Cov{A^n_0}\psi)$ is
    real-analytic as it is a polynomial in $\psi$ and $\nabla\psi$, and their complex conjugates.   We also note that $(M - z)^{-1}$ is
    complex-analytic in $z$ on the resolvent set of $M$, and therefore, $(M + |\psi|^2)^{-1}$ is analytic. (a) now follows.
\end{proof}

Now we substitute the expression \eqref{eq:a=} for $a$ into \eqref{eq:F-1} to get a single equation $F(\lambda, \psi) = 0$, where the map
$F : \R \times \Hpsi{2}{n} \to \Lpsi{2}{n}$ is defined as %$F(\lambda, \psi) := F_1(\lambda, \psi, a(\psi))$, or, explicitly,
\begin{equation}
\label{eq:F}
    F(\lambda, \psi) = (L^n - \lambda)\psi + 2ia(\psi)\cdot\nabla_{A^n_0}\psi + |a(\psi)|^2\psi + \kappa^2|\psi|^2\psi.
\end{equation}

The following proposition lists some properties of $F$.
\begin{proposition}
    \hfill
    \begin{enumerate}[(a)]
    \item $F$ is analytic as a map between real Banach spaces,
    \item for all $\lambda$, $F(\lambda, 0) = 0$,
    \item for all $\lambda$, $D_\psi F(\lambda, 0)=L^n - \lambda, $
    \item for all $\alpha \in \R$, $F(\lambda, e^{i\alpha}\psi) = e^{i\alpha}F(\lambda, \psi)$.
	\item for all $\psi$, $\langle \psi, F(\lambda, \psi) \rangle \in \R$.
    \end{enumerate}
\end{proposition}
\begin{proof}
    The first property follows from the definition of $F$ and the corresponding analyticity of $a(\psi)$. (b) through (d) are straightforward calculations. For (e), we calculate that
   	\begin{equation*}
		\langle \psi, F(\lambda, \psi) \rangle
		= \langle \psi, (L^n - \lambda)\psi \rangle + 2i\int_{\Omega^\tau} \bar{\psi}\alpha(\psi)\cdot\nabla\psi
			+ 2\int_{\Omega^\tau} (\alpha(\psi)\cdot A_0)|\psi|^2
			+ \int_{\Omega^\tau} |\alpha(\psi)|^2 |\psi|^2
			+ \kappa^2 \int_{\Omega^\tau} |\psi|^4.
	\end{equation*}
	The final three terms are clearly real and so is the first because $L^n - \lambda$ is self-adjoint. For the second term we calculate the complex conjugate and see that
	\begin{equation*}
		\overline{ 2i\int_{\Omega^\tau} \bar{\psi}\alpha(\psi)\cdot\nabla\psi }
			= -2i\int_{\Omega^\tau} \psi\alpha(\psi)\cdot\nabla\bar{\psi}
			= 2i\int_{\Omega^\tau} (\nabla\psi \cdot \alpha(\psi))\bar{\psi},
	\end{equation*}
	where we have integrated by parts and used the fact that the boundary terms vanish due to the periodicity of the integrand and that $\Div a(\psi) = 0$. Thus this term is also real and (e) is established.
\end{proof}

   Now we reduce the equation $F(\lambda,\psi) = 0$ to an equation on the finite-dimensional subspace $\Null (L^n - n)$. To this end we use the standard method of Lyapunov-Schmidt reduction. Let %$\lambda_0 :=N$,
   $X:=\Hpsi{2}{n}$ and $Y:= \Lpsi{2}{N}$ and let $K = \Null (L^n - n)$. We let $P$ be the Riesz projection onto $K$, that is,
    \begin{equation}
        P := -\frac{1}{2\pi i} \oint_\gamma (L^n - z)^{-1} \,dz,
    \end{equation}
    where $\gamma \subseteq \C$ is a contour around $0$ that contains no other points of the spectrum of $L$.
    This is possible since $0$ is an isolated eigenvalue of $L$. $P$ is a bounded, orthogonal projection, and if we
    let $Z := \Null P$, then $Y = K \oplus Z$. We also let $Q := I - P$, and so $Q$ is a projection onto $Z$.

    The equation $F(\lambda,\psi) = 0$ is therefore equivalent to the pair of equations
    \begin{align}
        \label{BT:eqn1} &P F(\lambda, P\psi + Q\psi) = 0, \\
        \label{BT:eqn2} &Q F(\lambda, P\psi + Q\psi) = 0.
    \end{align}

    We will now solve \eqref{BT:eqn2} for $w = Q\psi$ in terms of $\lambda$ and $v = P\psi$. To do this, we introduce the map
    $G : \R \times K \times Z \to Z$ to be $G(\lambda, v, w) := QF(\lambda, v + w)$.    Applying the Implicit Function Theorem
    to $G$, we obtain a real-analytic function $w : \R \times K \to Z$, defined on a neighbourhood of $(n, 0)$,
    such that $w = w(\lambda, v)$ is a unique solution to $G(\lambda, v, w) = 0$, for $(\lambda, v)$ in that neighbourhood.
    We substitute this function into \eqref{BT:eqn1} and see that the latter equation %we are looking for $(\lambda, \psi)$
     in a neighbourhood of $(n, 0)$ is equivalent to the equations
     \begin{equation}
    \label{BT:psivlam}\psi =v+ w(\lambda, v)
    \end{equation} and %$PF(\lambda, P\psi + w(\lambda, P\psi)) = 0$.
       \begin{equation}
    \label{BT:bif-eqn}
        \gamma(\lambda, v):= PF(\lambda, v + w(\lambda, v)) = 0
    \end{equation}
    (the \emph{bifurcation equation}). Note that %We therefore define   the function
     $\gamma : \R \times K \to \C$. %is defined by
   % \begin{equation}
   % \label{BT:gamma}
    %    \gamma(\lambda, v) := PF(\lambda, v + w(\lambda, v)).
   % \end{equation}
    We have shown that in a neighbourhood of $(n, 0)$ in $\R \times X$, $(\lambda, \psi)$ solves $F(\lambda, \psi) = 0$
    if and only if $(\lambda, v)$, with $v = P\psi$, solves
    \eqref{BT:bif-eqn}.

 Finally we note that $\gamma$ inherits the symmetry of the original equation:
    \begin{lemma}
        For every $\alpha \in \R$, $\gamma(\lambda, e^{i\alpha}v) = e^{i\alpha} \gamma(\lambda, v)$.
    \end{lemma}
    \begin{proof}
        We first check that $w(\lambda, e^{i\alpha}v) = e^{i\alpha} w(\lambda, v)$. We note that by definition of $w$,
        $G(\lambda,  e^{i\alpha} v, w(\lambda, e^{i\alpha} v)) = 0$, but by the symmetry of $F$, we also have
        $G(\lambda,  e^{i\alpha} v,  e^{i\alpha} w(\lambda, v)) =  e^{i\alpha} G(\lambda,v, w(\lambda,v)) = 0$. The uniqueness of $w$
        then implies that $w(\lambda, e^{i\alpha} v) = e^{i\alpha} w(\lambda, v)$.
        We can now verify that
        \begin{equation*}
            \gamma(\lambda, e^{i\alpha} v) = PF(\lambda, e^{i\alpha} v + w(\lambda,  e^{i\alpha} v))
                = e^{i\alpha} PF(\lambda, v + w(\lambda, v) ) \rangle = e^{i\alpha}\gamma(\lambda,v).
        \qedhere
        \end{equation*}
    \end{proof}

Solving the bifurcation equation \eqref{BT:bif-eqn} is a subtle problem unless $n=1$. The latter case is tackled in the next section.

%%%%%%%%%%%%%%%%%%%%%%%%%%%%%%%%%%%%%%%%%%%%%%%%%%%%%%%%%%%%%%%%%%%%%%%%%%%%%%%%%%%%%%%%%%%%%
%%%%%%%%%%%%%%%%%%%%%%%%%%%%%%%%%%%%%%%%%%%%%%%%%%%%%%%%%%%
%%%%%%%%%%%%%%%%%%%%%%%%%%%%%%%%%%%%%%%%%%%%%%%%%%%%%%%%%
\section{Bifurcation Theorem. $n=1$}
\label{sec:bifurcation-analysis n=1}

In this section we look at the case $n = 1$, and look at solutions near the trivial solution. For convenience we drop the index $n = 1$ from the notation. We will see that as $b  = \frac{\kappa^2 }{\lambda}$ decreases past the critical value $b = \kappa^2$, a branch of non-trivial solutions bifurcates from the trivial solution. More precisely, we have the following result.

\begin{theorem}
\label{thm:bifurcation-resultN1}
    %There exists an $\epsilon > 0$ such that
    For every $\tau$ %and non-zero $\psi_0 \in \Hpsi{2}{N}$ satisfying $(L^N - 1)\psi_0 = 0$,
    there exists a branch, $ (\lambda_s, \psi_s, A_s)$, $s \in \C$ with $|s|^2 < \epsilon$ for some $\epsilon > 0$, of nontrivial solutions of  the rescaled Ginzburg-Landau equations \eqref{eq:rGL}, unique (apart from the trivial solution $(1,0,A_0)$) in a sufficiently small neighbourhood of $(1,0,A_0)$ in $\R \times \mathscr{H}(\tau) \times \vec{\mathscr{H}}(\tau)$, and s.t.
    \begin{equation*}
    \begin{cases}
        \lambda_s = 1 + g_\lambda(|s|^2), \\
        \psi_s = s\psi_0 + sg_\psi(|s|^2), \\
        A_s = A_0 + g_A(|s|^2),
    \end{cases}
    \end{equation*}
    where $(L - 1)\psi_0 = 0,\ g_\psi$ is orthogonal to $\Null(L - 1)$,  $g_\lambda : [0,\epsilon) \to \R$, $g_\psi : [0,\epsilon) \to \mathscr{H}(\tau)$, and $g_A : [0,\epsilon) \to \vec{\mathscr{H}}(\tau)$ are real-analytic functions such that $g_\lambda(0) = 0$, $g_\psi(0) = 0$, $g_A(0) = 0$ and  $g'_\lambda(0) > 0$. Moreover,
    \begin{align} \label{glambda'}
		g_\lambda'(0) = \left(\kappa^2 - \frac{1}{2} \right) \frac{\int_{\Omega^\tau} |\psi_0|^4}{\int_{\Omega^\tau} |\psi_0|^2 }
							+ \frac{1}{4\pi} \int_{\Omega^\tau} |\psi_0|^2.
	\end{align}
    %, as well as a neighbourhood of $(1,0,0)$ in $\R \times \Hpsi{2}{N} \times \HA{2}{\Div,0}$. %such that $u = (\lambda, \psi, A)$ in that neighbourhood solves if and only if either $u$ is a trivial solution, i.e. $\psi = 0$ and $A = A^N_0$, or there is some $s \in \C$ with $|s|^2 < \epsilon$ such that $u = (\lambda_s, \psi_s, A_s)$, where
    \end{theorem}

%%%%

\begin{proof} The proof of this theorem is a slight modification of a standard result from the bifurcation theory. It can be  found in Appendix \ref{sec:bifurcation-theorem},  Theorem \ref{thm:bifurcation-theorem-1}, %{thm:bifurcation-theorem-2},
whose hypotheses are satisfied by $F$ as shown above (see also \cite{Odeh, BGT}).
%%%%%%%%%%%%%%%%%%%%%%%%%%%%%%%%%%%
%%%%%%%%%%%%%%%%%%%%%%%%%%%%%%%%%
The latter theorem gives us a neighbourhood of $(1, 0)$ in $\R \times \mathscr{H}(\tau)$ such that the only non-trivial solutions are given by
\begin{equation*}
\begin{cases}
    \lambda_s =1+ g_\lambda(|s|^2), \\
    \psi_s = s\psi_0 + sg_\psi(|s|^2).
\end{cases}
\end{equation*}
Recall $a(\psi)$ defined in \eqref{eq:a=}. We now define %$\tilde{g}_A$ to be
%\begin{equation*}
    $\tilde{g}_A(s) = a(\psi_s),$ %t\psi_0 + tg_\psi(t^2)),
%\end{equation*}
which is real-analytic and satisfies
%\begin{equation*}
$    \tilde{g}_A(-t) = a(-\psi_t %(t\psi_0 + tg_\psi(t^2))
    ) = \tilde{g}_A(t),$
%\end{equation*}
and therefore is really a function of $t^2,\ g_A(t^2)$. Hence $A_s = A_0 + g_A(|s|^2)$.

	Finally, to prove \eqref{glambda'} we multiply the equation $F(\lambda, \psi)=0$ scalarly by $\psi_0$ and use that $L$ is self-adjoint and $(L - 1)\psi_0=0$ to obtain
\begin{equation*}
\langle \psi_0,( \lambda -1)\psi \rangle = 2i\langle \psi_0, a(\psi)\cdot\nabla_{A_0}\psi \rangle + \langle \psi_0, |a(\psi)|^2\psi \rangle + \kappa^2\langle\psi_0, |\psi|^2\psi \rangle.
\end{equation*}
Let $a_1:=g_A'(0)$. Substituting here the expansions obtained in the first part of the theorem, we find
  \begin{align} \label{glambda'1}
		g_\lambda'(0) \| \psi_0 \|^2 = 2i\int_{\Omega^\tau} \bar{\psi}_0 a_1\cdot\Cov{A_0}\psi_0 + \kappa^2\int_{\Omega^\tau} |\psi_0|^4.
	\end{align}
	In order to simplify this expression we first note that by differentiating \eqref{eq:F-2} w.r. to $|s|^2$ at $s=0$, we obtain
    \begin{equation*}
         \Curl^* \Curl a_1 = \Im(\bar{\psi}_0\Cov{A_0}\psi_0).
    \end{equation*}
    Now for the first term on the r.h.s. of \eqref{glambda'1}, taking the imaginary part of \eqref{glambda'1} we see that $\Re(\int_{\Omega^\tau} \bar{\psi}_0 a_1\cdot\Cov{A_0}\psi_0) = 0$ and therefore
    \begin{align*}
        2i \int_{\Omega^\tau} \bar{\psi}_0 a_1\cdot\Cov{A_0}\psi_0 \,dx
%            &= i \int_{\Omega^\tau} \bar{\psi}_0 a_1\cdot\Cov{A_0}\psi_0 \,dx
%            		+ i \int_{\Omega^\tau} \bar{\psi}_0 a_1\cdot\Cov{A_0}\psi_0 \,dx \\
%            &= i \int_{\Omega^\tau} a_1\cdot \bar{\psi}_0 \Cov{A_0}\psi_0 \,dx
%            		- i \int_{\Omega^\tau} a_1\cdot\psi_0\Cov{-A_0}\bar{\psi}_0 \,dx \\
            &= -2 \int_{\Omega^\tau} a_1 \cdot \Im( \bar{\psi}_0\Cov{A_0}\psi_0 ) \,dx \\
            &= -2 \int_{\Omega^\tau} a_1 \cdot \Curl^* \Curl a_1 \,dx \\
            &= -2 \int_{\Omega^\tau} (\Curl a_1)^2 \,dx .
         \end{align*}
	Here in the second step we integrated by parts.
    \DETAILS{By substituting the expansions of $\psi^\tau$ and $A^\tau$ into \eqref{eq:F-2} (or \eqref{eq:a=}), we obtain the expansion
    \begin{equation}
        \Curl^* \Curl A^n_0 + \mu( M a^\tau_1 - \Im\{ \bar{\psi}^\tau_1\Cov{A^n_0}\psi^\tau_1 \} ) + o(\mu^2)=0,
    \end{equation}
  where  $\mu := \kappa^2 - b$.    We therefore have that}
  %
  %
    %If we let $f := -\frac{1}{2}|\psi_0|^2$, a
%	To compute the integral on the r.h.s.
	Next we show that
	\begin{equation}
	\label{eq:curl-psi_0}
		-\frac{1}{2}\Curl^* |\psi_0|^2 = \Im(\bar{\psi_0}\Cov{A_0}\psi_0).
	\end{equation}
	Using the notations of Appendix \ref{sec:operators} (with $n=1$), we have that $L_-\psi_0 = 0$:
	\begin{equation*}
		\partial_{x_1}\psi_0 + i\partial_{x_2}\psi_0 + \frac{1}{2}x_1\psi_0 + \frac{i}{2}x_2\psi_0 = 0
	\end{equation*}
	Multiplying this relation by $\bar{\psi}_0$ and subtracting and adding the complex conjugate of the result, we obtain the two relations
	\begin{equation*}
		\begin{cases}
			\bar{\psi_0}\partial_{x_1}\psi_0 - \psi_0\partial_{x_1}\bar{\psi}_0
					= -i\bar{\psi}_0\partial_{x_2}\psi_0 - i\psi_0\partial_{x_2}\bar{\psi}_0 - i x_2|\psi_0|^2, \\
			\bar{\psi_0}\partial_{x_2}\psi_0 - \psi_0\partial_{x_2}\bar{\psi}_0
					= i\bar{\psi}_0\partial_{x_1}\psi_0 + i\psi_0\partial_{x_1}\bar{\psi}_0 + i x_1|\psi_0|^2.
		\end{cases}
	\end{equation*}
	This means that
	\begin{align*}
		\Im(\bar{\psi_0}\Cov{A_0}\psi_0)
		&= \left(\begin{array}{c}
			-\frac{i}{2}( \bar{\psi_0}\partial_{x_1}\psi_0 - \psi_0\partial_{x_1}\bar{\psi}_0 ) + \frac{n}{2}x_2|\psi_0|^2 \\
			-\frac{i}{2}( \bar{\psi_0}\partial_{x_2}\psi_0 - \psi_0\partial_{x_2}\bar{\psi}_0 ) - \frac{n}{2}x_1|\psi_0|^2
			\end{array} \right) \\
		&= \left(\begin{array}{c}
			-\frac{1}{2}(\bar{\psi}_0\partial_{x_2}\psi_0 + \psi_0\partial_{x_2}\bar{\psi}_0) \\
			\frac{1}{2}(\bar{\psi}_0\partial_{x_1}\psi_0 + \psi_0\partial_{x_1}\bar{\psi}_0 )
			\end{array} \right) \\
		&= \left(\begin{array}{c}
			-\frac{1}{2} \partial_{x_2} |\psi_0|^2 \\
			\frac{1}{2} \partial_{x_1} |\psi_0|^2
			\end{array} \right),
	\end{align*}
	which gives \eqref{eq:curl-psi_0}.
	
	\eqref{eq:curl-psi_0} implies that $\Curl a_1 = -\frac{1}{2}|\psi_0|^2 + C$
    for some constant $C$. This $C$ can be determined using the fact that, since $A$ has mean zero, $a_1$ does as well, and this gives
    $C = \frac{1}{4\pi}\int_{\Omega^\tau} |\psi_0|^2$, which establishes %\eqref{ABR:curlA1}
    \begin{equation}
    \label{ABR:curlA1}
        \Curl a_1 = -\frac{1}{2}|\psi_0|^2 + \frac{1}{4\pi}\int_{\Omega^\tau} |\psi_0|^2 .
    \end{equation}
 Using this equation we finish the calculation above:
 \begin{align} \label{intaimpsinablapsi}
		 2i \int_{\Omega^\tau} \bar{\psi}_0 a_1\cdot\Cov{A_0}\psi_0 \,dx
		 = - \frac{1}{2} \int_{\Omega^\tau} |\psi_0|^4 \,dx + \frac{1}{4\pi} \left( \int_{\Omega^\tau} |\psi_0|^2 \,dx \right)^2.
%        \int_{\Omega^\tau} \bar{\psi}_0 (2ia_1\cdot\Cov{A^n_0}\psi_0) \,dx
                        %&= -2 \int_{\Omega^\tau} \left( -\frac{1}{2}|\psi_0|^2 + \frac{1}{4\pi}\int_{\Omega^\tau} |\psi_0|^2 \,dx \right)^2 \,dx \\
%            = -\frac{1}{2} \int_{\Omega^\tau} |\psi_0|^4 \,dx + \frac{1}{4\pi} \left( \int_{\Omega^\tau} |\psi_0|^2 \,dx \right)^2.
    \end{align}
 Substituting this expression into \eqref{glambda'1} and rearranging terms we arrive at \eqref{glambda'}.
 And that completes the proof of Theorem \ref{thm:bifurcation-resultN1}.
\end{proof}
Theorem \ref{thm:bifurcation-resultN1} implies (I) - (III) of Theorem \ref{thm:result}.
$\qed$

Finally, we mention
\begin{lemma} \label{lem:tauanal}
        	%We can prove that this
Recall that $\Im\tau >0$. Let  $ (\lambda_s, \psi_s, A_s)$  be the solution branch constructed above and let  %$M_\tau$ be the matrix
$m_\tau =\\ (\sqrt{\Im\tau})^{-1}\left( \begin{array}{cc} 1 & \Re\tau \\ 0 & \Im\tau \end{array} \right)$. Then  $ (\lambda_s, \tilde{\psi}_s, \tilde{A}_s)$, where the functions $ ( \tilde{\psi}_s, \tilde{A}_s)$ are defined on a $\tau$-independent square lattice and are given by
      %in the following way. Instead of the rescaling of Section \ref{sec:reduced-problem}, we rescale pairs $(\psi, A)$ as follows:
	\begin{equation} \label{Utau}
		\begin{cases}
		\tilde{\psi}_s(x) =  \psi_s( m_\tau x), \\
		\tilde{A}_s(x) =  M_\tau^t A_s( m_\tau x), \\
		\end{cases}
	\end{equation}
%where and $\sigma$ is as before defined by $\sigma =\left(\frac{n}{b}\right)^{\frac{1}{2}}$.
 depend $\R$-analytically on $\tau$ .
    \end{lemma}
 %   \begin{proof}
We sketch the proof of this lemma. The transformation above maps functions on a lattice of the shape $\tau$ into functions on a $\tau$-independent square lattice, but leads to a slightly more complicated expression for the Ginzburg-Landau equations. Namely, let  $U_\tau \psi (x):= \psi( m_\tau x)$ and  $V_\tau a(x):= m_\tau^t a( m_\tau x)$. Applying $U_\tau$ and $V_\tau$ to the equations \eqref{eq:F-1} and \eqref{eq:F-2}, we conclude that $ ( \tilde{\psi}_s, \tilde{A}_s)$ satisfy the equations
\begin{subequations}
    \begin{equation}
    \label{eq:F-1'}
        %F_1(\lambda, \psi, a) :=
        (L^n_\tau - \lambda)\psi
        + 2i (m_{\tau}^{t})^{-1}a\cdot (m_{\tau}^{t})^{-1}\nabla_{A^n_0}\psi %- 2A^n_0\cdot P\psi
        + |(m_{\tau}^{t})^{-1}a|^2\psi + \kappa^2|\psi|^2\psi =0,
        %+ \tilde{F}^\psi_\tau(\lambda, \psi, a) = 0,
    \end{equation}
    \begin{equation}
    \label{eq:F-2'}
        %F_2(\psi, a) :=
        (M_\tau + |\psi|^2)a
        %- \Im\{ \bar{\psi}\Cov{A^n_0}\psi \}=0,
        + \tilde{F}^a_\tau(\psi) = 0,
    \end{equation}
\end{subequations}
where %, with $U_\tau$, the operator defined by \eqref{Utau} (it acts differently of $\psi$ and $A$),
\begin{equation}
\label{OP:LN'}
    L^n_\tau := -U_\tau\Delta_{A^n_0}U_\tau^{-1}
%\end{equation}
\mbox{ and }
%\begin{equation}
%\label{OP:M}
    M_\tau := V_\tau\Curl^*\Curl V_\tau^{-1}.
\end{equation}
Here we used that $V_\tau A^n_0= A^n_0$ and $U_\tau \nabla \psi = (m_{\tau}^{t})^{-1}U_\tau  \psi$. (The latter relation is a straightforward computation and the former one follows from the facts that for any matrix $m,\ (mx)^\perp = (\det m) (m^t)^{-1} x^\perp$, and that in our case, $\det m_\tau = 1$.) Note that the gauge in the periodicity condition will still depend on $\Im\tau$. These complications, however, are inessential and the same techniques as above can be applied in this case. The important point here is to observe that the function $\psi_0$, constructed in Appendix B, the function $w(\lambda, s\psi_0)$, where $w(\lambda, v)$ is  the solution of \eqref{BT:eqn2}, and the bifurcation equation \eqref{BT:bif-eqn} depend on $\tau$ real-analytically. We leave the details of the proof to the interested reader.
%\end{remark}

%%%%%%%%%%%%%%%%%%%%%%%%%%%%%%%%%%%%%%%%%%%%%%%%%%%%%%%%%%%%%%%%%%%%%%%%%%%%%%%%%%%%%%%%%%%%%
%%%%%%%%%%%%%%%%%%%%%%%%%%%%%%%%%%%%%%%%%%%%%%%%%%%%%%%%%%%%%%%%%%%%%%%%%%%%%%%%%%%%%%%%%%%%%

\section{Abrikosov Function}
\label{S:ABR}

In this section, we continue with the case $n = 1$. We prove the Abrikosov relation between the energy per cell and the Abrikosov function, $\beta(\tau)$, which is defined by
\begin{equation} \label{beta}
    \beta(\tau) := \frac{ \int_{\Omega^\tau} |\psi^\tau_0|^4 }{ \left( \int_{\Omega^\tau} |\psi^\tau_0|^2 \right)^2 },
\end{equation}
where $\psi^\tau_0$ is a non-zero element in the nullspace of  $L^n - 1$ acting on $\Hpsi{2}{n}$. Since the nullspace is a one-dimensional complex
subspace, $\beta$ is well-defined.

Recall that $b= \frac{\kappa^2}{\lambda}.$ % Theorem \ref{thm:bifurcation-resultN1}  implies the existence of a function $g_b(|s|)$, s.t. $b_s=\kappa^2 +g_b(|s|), $ with $g_b(0) = 0$ and $g_b'(0) \neq 0$. The latter two properties imply that the function $b =\kappa^2 + g_b(|s|)$
Since the function $g_\lambda(|s|^2)$ given in Theorem \ref{thm:bifurcation-resultN1} obeys  $g_\lambda(0) = 0$ and $g_\lambda '(0) \neq 0$, the function $b_s = \kappa^2 (1 + g_\lambda(|s|^2))^{-1} =: \kappa^2 + g_b(|s|^2)$
can be inverted to obtain %, with the help of the relation $b= \frac{\kappa^2}{\lambda}$, that
$|s| = s(b)$. Absorbing $\hat{s} = \frac{s}{|s|}$ into $\psi^\tau_0$, we can define the family $(\psi^\tau_{s(b)}, A^\tau_{s(b)}, b^\tau_{s(b)})$ of $\cL^\tau$-periodic solutions of the Ginzburg-Landau equations parameterized by average magnetic flux $b$ and their energy
$$E_b(\tau) := \mathcal{E}_{\kappa^2/b}(\psi^\tau_{s(b)}, A^\tau_{s(b)}).$$ %where $s(b) := \hat{s}g_b^{-1}(|s|)$, $\hat{s} = \frac{s}{|s|}$.
Clearly, $\psi^\tau_{s(b)}, A^\tau_{s(b)}, b^\tau_{s(b)}$ are analytic in $b$. %Though $(\psi^\tau_s, A^\tau_s, b^\tau_s)$ is fixed up to multiplication of $\psi^\tau_s$ by a factor of $e^{i\alpha}$, the energy $E_b(\tau)$ is unique.
We note the relation between the new perturbation parameter $\mu := \kappa^2 - b$ and the bifurcation parameter $|s|^2:$ \begin{equation}
    \label{mus}
         \mu= \frac{g_\lambda(|s|^2)}{\lambda}\kappa^2= g'_\lambda(0)\kappa^2 |s|^2 +O(|s|^4).
    \end{equation}
The relation between the Abrikosov function and the energy of the Abrikosov
lattice solutions is as follows.

\begin{theorem}
\label{ABR:thm1}
    In the case $\kappa > \frac{1}{\sqrt{2}}$, the minimizers,  $\tau_b$, of $\tau \mapsto E_b(\tau)$ are related to the minimizer,  $\tau_*$, of $\beta(\tau)$, as $\tau_b - \tau_* =O(\mu^{1/2})$, In particular, $\tau_b \to \tau_*$ as $b \to \kappa^2$.
    %for all $b$ sufficiently close to but less than $\kappa^2$, $\tau$ minimizes $E_b(\tau)$ if and only if
    %$\tau$ minimizes $\beta(\tau)$.
\end{theorem}
\begin{proof}
We first show that the theorem is a consequence of the following proposition, which is proved below.
\begin{proposition}
\label{prop:Ebeta}
	We have
            \begin{equation}
    \label{Ebeta}
        E_b(\tau) = \frac{\kappa^2}{2} + \kappa^4 - 2\kappa^2\mu +
        	\left( \frac{\kappa^4}{4\pi} - \frac{1}{1 + 4\pi(\kappa^2 - \frac{1}{2})\beta(\tau)} \right)\mu^2 + O(\mu^3).
    \end{equation}
    \end{proposition}

%	Let $\tau_b$ be the minimizer of $E_b(\tau)$ and $\tau_*$ the minimizer of $\beta(\tau)$. We show that $\tau_b - \tau_* =O(\mu^{1/2}).$
%$\tau_b \to \tau_*$ as $b \to b_c$.
To prove the theorem we note that  $E_b(\tau)$ is of the form $E_b(\tau) = e_0 + e_1 \mu + e_2(\tau) \mu^2 + O(\mu^3)$. The first two terms are constant in $\tau$, so we consider $\tilde{E}_b(\tau) = e_2(\tau) + O(\mu)$. $\tau_b$ is also the minimizer of $\tau \mapsto \tilde{E}_b(\tau)$ and $\tau_*$, of $e_2(\tau)$. We have the expansions  $\tilde{E}_b(\tau_*)-\tilde{E}_b(\tau_b) = \frac{1}{2}\tilde{E}^{''}_b(\tau_b)(\tau_*-\tau_b)^2 + O((\tau_*-\tau_b)^3)$ and  $\tilde{E}_b(\tau_*)-\tilde{E}_b(\tau_b) = -\frac{1}{2}e^{''}_2(\tau_b)(\tau_*-\tau_b)^2 + O((\tau_*-\tau_b)^3) + O(\mu)$, which imply the desired result. That concludes the proof of the theorem.
\end{proof}
\DETAILS{    Now suppose that there is a sequence $b_k \to b_c$ such that $\tau_{b_k} \to \sigma$ with $\sigma \neq \tau_*$. Since $\tilde{E}_b(\tau_{b_k}) \to e_2(\sigma)$, for any $\epsilon > 0$, $\tilde{E}_{b_k}(\tau_{b_k}) \geq e_2(\sigma) - \epsilon$ for sufficiently large $k$. By continuity we also know that for all $\epsilon' > 0$, $\tilde{E}_b(\tau_*) \leq e_2(\tau_*) + \epsilon'$ for $b$ sufficiently close to $b_c$. This implies that for any $\epsilon, \epsilon' > 0$, for all sufficiently large $k$ we have
    \begin{equation*}
    	\tilde{E}_{b_k}(\tau_*) \leq e_2(\tau_*) + \epsilon'
    		\leq e_2(\sigma) + (e_2(\tau_*) - e_2(\sigma)) + \epsilon'
    		\leq \tilde{E}_{b_k}(\tau_{b_k}) + (e_2(\tau_*) - e_2(\sigma)) + \epsilon + \epsilon'.
    \end{equation*}
    But by the minimality of $\tau_star$, $e_2(\tau_*) - e_2(\sigma) < 0$, so for sufficiently small $\epsilon, \epsilon'$, we would have $\tilde{E}_{b_k}(\tau_*) < \tilde{E}_{b_k}(\tau_{b_k})$, contradicting the minimality of $\tau_{b_k}$.

    Thus $\tau_b \to \tau_*$ as $b \to b_c$ but it is clear from \eqref{Ebeta} that if $\kappa^2 > 1/2$, then $\tau_*$ minimizes $e_2(\tau)$ if and only if $\tau_*$ minimizes $\beta(\tau)$.}

\begin{proof}[Prof of Proposition \ref{prop:Ebeta}]
    %We expand $|s|^2$, near $|s| = 0$ in terms of $\mu := \kappa^2 - b$, as $|s|^2 = c_1 \mu +c_2 \mu^2 + \dots.$
 Recall that  $\mu := \kappa^2 - b$.    Using the real-analyticity of the function $g_b$, $g^\tau_\psi$, and $g^\tau_A$, we can express $\lambda(\mu) := \kappa^2 / b$, $\psi^\tau(\mu) := \psi^\tau_{s(b)}$ and
    $A^\tau(\mu) := A^\tau_{s(b)}$ as
    \begin{align}
    	&\lambda(\mu) = 1 + \frac{1}{\kappa^2} \mu + O(\mu^2) \\
        \label{Exppsimu}&\psi^\tau(\mu) = \mu^{1/2}\psi^\tau_0 + \mu^{3/2}\psi^\tau_1 + O(\mu^{5/2}) \\
        &A^\tau(\mu) = A_0 + \mu a^\tau_1 + O(\mu^2). \label{Expamu}
    \end{align}
    We will first show that
    \begin{equation}
    \label{ABR:Ebomu}
   	    E_b(\tau) = \frac{\kappa^2}{2} + \kappa^4
				 	- 2\kappa^2 \mu
	    	+ \frac{\kappa^4}{4\pi} \left( 1 - \frac{1}{\kappa^2} \int_{\Omega^\tau} |\psi_0^\tau|^2  \right) \mu^2 + O(\mu^3).
    	    %E_b(\tau) = \frac{\kappa^2}{2} + \kappa^4 - 2\kappa^2 \mu + \frac{\kappa^4}{4\pi} \left( 1 - \left(\kappa^2 - \frac{1}{2}\right) \int |\psi_0^\tau|^4 - \frac{1}{4\pi}\left(\int |\psi^\tau_0|^2 \right)^2 \right) \mu^2 + O(\mu^4).
        %E_b(\tau) = \frac{\kappa^2}{2} + \mu^2 \left( 1 - \kappa^2 \int_{\Omega^\tau} |\psi^\tau_0|^2 \right) + o(\mu^4).
    \end{equation}
    %%%%%%%%%%%%%%%%%%%%%%%%%%%%%%%%%
   Multiplying \eqref{rGLpsi} scalarly by $\psi$ and integrating by parts gives
   \begin{equation*}
  		\int_{\Omega^\tau} |\nabla_A \psi^\tau|^2 = \kappa^2 \int_{\Omega^\tau} \left(\lambda|\psi^\tau|^2 - \kappa^2|\psi^\tau|^4\right).
	\end{equation*}
   Substituting this into the expression for the energy we find
   \begin{equation}
   \label{ABR:Ebomu'}
		E_b(\tau) = \frac{\kappa^4}{2\pi\lambda^2} \int_{\Omega^\tau} \left(\frac{\lambda^2}{2\kappa^2} - \frac{\kappa^2}{2} |\psi^\tau|^4 + |\Curl A^\tau|^2\right).
        %E_b(\tau) = \frac{\kappa^2}{2}  -\frac{\kappa^2}{2}\langle|\psi|^4\rangle +  \langle|\Curl A|^2\rangle.
    \end{equation}
    %where $\langle f \rangle :=\frac{1}{|\Omega|}\int f$.
    Using the expansion above gives
     \begin{equation}
    \label{ABR:Ebomu''}
	    E_b(\tau) = \frac{\kappa^2}{2} + \kappa^4
				 	- 2\kappa^2 \mu
	    	+ \frac{\kappa^4}{2\pi} \left( 1 - \frac{\kappa^2}{2} \int_{\Omega^\tau} |\psi_0^\tau|^4 + \int_{\Omega^\tau} |\Curl a_1^\tau|^2 \right)\mu^2
	    		+ O(\mu^3),
        %E_b(\tau) = \frac{\kappa^2}{2}  -\mu^2\left(\frac{\kappa^2}{2}\langle|\psi^\tau_0|^4\rangle -  \langle|\Curl a^\tau_1|^2\rangle \right) +o(\mu^2).
    \end{equation}
    where we have used the fact that $\Curl A_0 = 1$.
    As we proved above in \eqref{ABR:curlA1}, we can show that
    \begin{equation*}
    	\Curl a^\tau_1 = -\frac{1}{2}|\psi^\tau_0|^2 + \frac{1}{4\pi} \int_{\Omega^\tau} |\psi^\tau_0|^2.
    \end{equation*}
    Substituting this expression into \eqref{ABR:Ebomu''}, we obtain
    \begin{equation}
    \label{ABR:Ebomu'''}
    	    E_b(\tau) = \frac{\kappa^2}{2} + \kappa^4
				 	- 2\kappa^2 \mu
	    	+ \frac{\kappa^4}{4\pi} \left( 1 - \left(\kappa^2 - \frac{1}{2}\right) \int_{\Omega^\tau} |\psi_0^\tau|^4
	    			- \frac{1}{4\pi}\left(\int_{\Omega^\tau} |\psi^\tau_0|^2 \right)^2 \right) \mu^2 + O(\mu^3).
        %E_b(\tau) = \frac{\kappa^2}{2} + \mu^2 \left( 1 - \kappa^2 \int_{\Omega^\tau} |\psi^\tau_0|^2 \right) + o(\mu^4).
    \end{equation}

    Now, if we differentiate  \eqref{eq:F-1} twice w.r. to $\mu^{1/2}$ at $\mu=0$ to obtain
    \begin{equation*}
        (L - 1)\psi^\tau_1 = \frac{1}{\kappa^2}\psi^\tau_0 - \kappa^2|\psi^\tau_0|^2\psi^\tau_0 - 2ia^\tau_1\cdot\Cov{A_0}\psi^\tau_0.
    \end{equation*}
    Now using the fact that $L - 1$ is self-adjoint and that $(L - 1)\psi^\tau_0 = 0$, we find
    \begin{align*}
        0 &= \int_{\Omega^\tau} \bar{\psi}^\tau_0(L - 1)\psi^\tau_1 \,dx \\
            &= \int_{\Omega^\tau} \bar{\psi}^\tau_0(\frac{1}{\kappa^2}\psi^\tau_0 - \kappa^2|\psi^\tau_0|^2\psi^\tau_0 - 2ia^\tau_1\cdot\Cov{A_0}\psi^\tau_0) \,dx \\
            &= \frac{1}{\kappa^2} \int_{\Omega^\tau} |\psi^\tau_0|^2 \,dx - \kappa^2 \int_{\Omega^\tau} |\psi^\tau_0|^4 \,dx
                - \int_{\Omega^\tau} \bar{\psi}^\tau_0 (2ia^\tau_1\cdot\Cov{A_0}\psi^\tau_0) \,dx.
    \end{align*}
\DETAILS{For the third term, we integrate by parts and note that boundary terms disappear because of their periodicity as well as noting that
    $\Div A^\tau_1 = 0$ to calculate that
    \begin{align*}
        \int_{\Omega^\tau} \bar{\psi}^\tau_0 (2iA^\tau_1\cdot\Cov{A^n_0}\psi^\tau_0) \,dx
            &= -2 \int_{\Omega^\tau} A^\tau_1 \cdot \Im( \bar{\psi}^\tau_0\Cov{A^n_0}\psi^\tau_0 ) \,dx \\
            &= -2 \int_{\Omega^\tau} A^\tau_1 \cdot \Curl^* \Curl A^\tau_1 \,dx \\
            &= -2 \int_{\Omega^\tau} (\Curl A^\tau_1)^2 \,dx \\
            &= -2 \int_{\Omega^\tau} \left( -\frac{1}{2}|\psi^\tau_0|^2 + \frac{1}{4\pi}\int_{\Omega^\tau} |\psi^\tau_0|^2 \,dx \right)^2 \,dx \\
            &= -\frac{1}{2} \int_{\Omega^\tau} |\psi^\tau_0|^4 \,dx + \frac{1}{4\pi} \left( \int_{\Omega^\tau} |\psi^\tau_0|^2 \,dx \right)^2.
    \end{align*}
    Therefore}
    %
    %
%   Now, notice that due to \eqref{twoexp}, Eqn \eqref{intaimpsinablapsi} still holds if we replace $\psi_0,\ a_1$ by $\psi^\tau_0,\ %a^\tau_1$. Using this together with the last relation gives
	An analogous calculation to the one in the proof of Theorem \ref{thm:bifurcation-resultN1} then gives
    \begin{equation*}
        0 = \frac{1}{\kappa^2} \int_{\Omega^\tau} |\psi^\tau_0|^2 \,dx - \kappa^2 \int_{\Omega^\tau} |\psi^\tau_0|^4 \,dx
                + \frac{1}{2} \int_{\Omega^\tau} |\psi^\tau_0|^4 \,dx  - \frac{1}{4\pi} \left( \int_{\Omega^\tau} |\psi^\tau_0|^2 \,dx \right)^2,
    \end{equation*}
    This relation gives \eqref{ABR:Ebomu} from \eqref{ABR:Ebomu'''}, but also by dividing by $\left( \int |\psi^\tau_0|^2 \right)^2$ and rearranging we then obtain \eqref{Ebeta}.
    \end{proof}
    %This implies that when $\kappa^2 > \frac{1}{2}$, in order to minimize the right hand side, we need to minimize $\beta$ and that proves the theorem.

The following result was discovered numerically in the physics literature and proven in \cite{ABN} using earlier result of \cite{NV}:

\begin{theorem}
\label{ABR:thm2}
    The function $\beta(\tau)$ has exactly two critical points, $\tau = e^{i\pi/3}$ and $\tau = e^{i\pi/2}$. The first is minimum, whereas the
    second is a maximum.
\end{theorem}

Theorems \ref{thm:bifurcation-resultN1}, %{BIF:thm},
\ref{ABR:thm1}, \ref{ABR:thm2}, after rescaling to the original variables, imply Theorem \ref{thm:result}, which, as was mentioned above, a precise restatement of Theorem \ref{thm:main-result} of Introduction.

%%%%%%%%%%%%%%%%%%%%%%%%%%%%%%%%%%%%%%%%%%%%%%%%%%%%%%%%%%%%%%%%%%%%%%%%%%%%%%%%%%%%%%%%%%%%%

%\section{Discussion: Stability}
%\label{sec:stability}
%
%
%We say an Abrikosov lattice $(\psi_\tau, A_\tau)$ is stable if its linearized operator is nonnegative  $L_{\psi_\tau, A_\tau} \ge 0$ with zero modes due to the symmetries. Stability of Abrikosov lattices? See \cite{FHP}. Use that $L_{\psi_\tau, A_\tau}$ is translationally invariant (or equivariant) w.r.to $\cL_\tau$ to find a block-diagonal decomposition of  $L_{\psi_\tau, A_\tau}$ (cf \cite{OS, GS}).
%
%%%%%%%%%%%%%%%%%%%%%%%%%%%%%%%%%%%%%%%%%%%%%%%%%%%%%%%%%%%%%%%%%%%%%%%%%%%%%%%%%%%%%%%%%%%%%

\appendix

%%%%%%%%%%%%%%%%%%%%%%%%%%%%%%%%%%%%%%%%%%%%%%%%%%%%%%%%%%%%%%%%%%%%%%%%%%%%%%%%%%%%%%%%%%%%%

\section{Bifurcation with Symmetry}
\label{sec:bifurcation-theorem}
In this appendix we present a variant of a standard result in Bifurcation Theory.
\begin{theorem}
\label{thm:bifurcation-theorem-1}
    Let $X$ and $Y$ be complex Hilbert spaces, with $X$ a dense subset of $Y$, and consider a map $F : \R \times X \to Y$ that
    is analytic as a map between real Banach spaces. Suppose that for some $\lambda_0 \in \R$, the following conditions are satisfied:
    \begin{enumerate}
        \item \label{BT:trivial} $F(\lambda, 0) = 0$ for all $\lambda \in \R$,
        \item \label{BT:L-0} $D_\psi F(\lambda_0, 0)$ is self-adjoint and has an isolated eigenvalue at $0$ of (geometric) multiplicity $1$,
%        \item \label{BT:v-star} For non-zero $v^* \in \Null (D_\psi F(\lambda_0, 0)^*)$ and $v \in \Null D_\psi F(\lambda_0, 0)$,
%                    $\langle v^*, D_{\lambda,\psi}F(\lambda_0, 0)v \rangle \neq 0$,
		\item \label{BT:v-star} For non-zero $v \in \Null D_\psi F(\lambda_0, 0)$,
                    $\langle v, D_{\lambda,\psi}F(\lambda_0, 0)v \rangle \neq 0$,
        \item \label{BT:sym} For all $\alpha \in \R$, $F(\lambda, e^{i\alpha}\psi) = e^{i\alpha}F(\lambda, \psi)$.
		\item \label{BT:SA} For all $\psi \in X$, $\langle \psi, F(\lambda, \psi) \rangle \in \R$.
    \end{enumerate}
    Then $(\lambda_0, 0)$ is a bifurcation point of the equation $F(\lambda, \psi) = 0$. In fact, there is a family of non-trivial solutions, $(\lambda, \psi)$, unique in a neighbourhood of $(\lambda_0, 0)$ in $\R \times X$, and this family has the form
    \begin{equation*}
    \begin{cases}
        \lambda = \phi_\lambda(|s|^2), \\
        \psi = sv + s\phi_\psi(|s|^2),
    \end{cases}
    \end{equation*}
    for $s \in \C$ with $|s| < \epsilon$, for some $\epsilon >0$. Here $v \in \Null D_\psi F(\lambda_0, 0)$, %there exists an $\epsilon > 0$
     and  $\phi_\lambda : [0,\epsilon) \to \R$ and $\phi_\psi : [0,\epsilon) \to X$ are unique real-analytic functions,
    such that $\phi_\lambda(0) = \lambda_0$, $\phi_\psi(0) = 0$. %and there is a neighbourhood of $(\lambda_0, 0)$ in $\R \times X$ such that
\end{theorem}
\begin{proof}%[Proof of Theorem \ref{thm:bifurcation-theorem-1}]
The analysis of Section \ref{sec:bifurcation-analysis n=1} reduces the problem to the one of solving % We have solve
the bifurcation equation \eqref{BT:bif-eqn}. Since the projection $P$, defined there, is rank one and self-adjoint, we have
%, for any non-zero $v_0 \in K$, we can find $v_0^\star \in \Null L^*$, such that $\langle v_0^\star, v_0 \rangle = 1$ and
    \begin{equation}
    \label{BT:P}
        P\psi = \frac{1}{\|v\|^2} \langle v, \psi \rangle v,\ \textrm{with}\ v \in \Null D_\psi F(\lambda_0, 0).
    \end{equation}
	We can therefore view the function $\gamma$ in the bifurcation equation \eqref{BT:bif-eqn} as a map $\gamma : \R \times \C \to \C$, where
	\begin{equation*}
		\gamma(\lambda, s) = \langle v, F(\lambda, sv_0 + w(\lambda, sv) \rangle.
	\end{equation*}
    We now look for non-trivial solutions of this equation, by using the Implicit Function Theorem to solve for $\lambda$ in terms of $s$.
    % Because of the symmetry of $\gamma$, we let $\gamma_0 : \R \times \R \to \C$ be the restriction of $\gamma$ to $\R \times \R$.
%    Note that if $\gamma_0(\lambda, t) = 0$, then $\gamma(\lambda, e^{i\alpha}t) = 0$ for all $\alpha$, and conversely, if $\gamma(\lambda, s) = 0$, then $\gamma_0(\lambda, |s|) = 0$. So we need only to find solutions of $\gamma_0(\lambda, t) = 0$.
	Note that if $\gamma(\lambda, t) = 0$, then $\gamma(\lambda, e^{i\alpha}t) = 0$ for all $\alpha$, and conversely, if $\gamma(\lambda, s) = 0$, then $\gamma(\lambda, |s|) = 0$. So we need only to find solutions of $\gamma(\lambda, t) = 0$ for $t \in \R$. We now show that $\gamma(\lambda, t) \in \R$.
	Since the projection $Q$ is self-adjoint, and since $Qw(\lambda, v) = w(\lambda, v)$ we have
	\begin{align*}
		\langle w(\lambda, tv), F(\lambda, tv+ w(\lambda, tv) \rangle
		= \langle w(\lambda, tv), QF(\lambda, tv + w(\lambda, tv) \rangle
		= 0.
	\end{align*}
	Therefore, for $t \neq 0$,
	\begin{equation*}
		\langle v, F(\lambda, tv + \Phi(\lambda, tv)) \rangle
		= t^{-1}\langle tv + w(\lambda, tv), F(\lambda, tv + w(\lambda, tv)) \rangle,
	\end{equation*}
	and this is real by condition \eqref{BT:SA} of the theorem. Thus we can restrict $\gamma$ to a function $\gamma_0 : \R \times \R \to \R$.

    By a standard application of  the Implicit Function Theorem to $t^{-1}\gamma_0(\lambda, t) = 0$, in which \eqref{BT:trivial}-\eqref{BT:v-star} are used (see for example \cite{AmPr}), there is $\epsilon > 0$ and a real-analytic function
    $\widetilde{\phi}_\lambda : (-\epsilon, \epsilon) \to \R$ such that $\widetilde{\phi}_\lambda(0) = \lambda_0$ and if $\gamma_0(\lambda, t) = 0$ with $|t| < \epsilon$,
    then either $t = 0$ or $\lambda = \phi_\lambda(t)$.     Recalling that $\gamma(\lambda, e^{i\alpha}t) = e^{i\alpha}\gamma(\lambda,t)$, we have shown that
    if $\gamma(\lambda, s) = 0$ and $|s| < \epsilon$, then either $s = 0$ or $\lambda = \phi_\lambda(|s|)$.

    We also note that by the symmetry, $\widetilde{\phi}_\lambda(-t) = \widetilde{\phi}_\lambda(|t|) = \widetilde{\phi}_\lambda(t)$, so $\widetilde{\phi}_\lambda$ is an even real-analytic
    function, and therefore must in fact be a function solely of $|t|^2$. We therefore set $\phi_\lambda(t) = \widetilde{\phi}_\lambda(\sqrt{t})$, and
    so $\phi_\lambda$ is real-analytic.

    We now define $\phi_\psi : (-\epsilon, \epsilon) \to \R$ to be
    \begin{equation}
        \phi_\psi(t) = \begin{cases}
                        t^{-1} w(\phi_\lambda(t), tv) &  t \neq 0, \\
                        0 & t = 0,
                    \end{cases}
    \end{equation}
    $\phi_\psi$ is also real-analytic and satisfies $s\phi_\psi(|s|^2) = w(\phi_\lambda(|s|^2), sv)$ for any $s \in \C$ with $|s|^2 < \epsilon$.

    Now we know that there is a neighbourhood of $(\lambda_0, 0)$ in $\R \times \Null D_\psi F(\lambda_0, 0)$ such that in that neighbourhood $F(\lambda, \psi) = 0$
    if and only if $\gamma(\lambda, s) = 0$ where $P\psi = sv$. By taking a smaller neighbourhood if necessary, we have proven that
    $F(\lambda, \psi) = 0$ in that neighbourhood if and only if either $s = 0$ or $\lambda = \phi_\lambda(|s|^2)$. If $s = 0$, we have
    $\psi = sv + s\phi_\psi(|s|^2) = 0$ which gives the trivial solution. In the other case, $\psi = sv + s\phi_\psi(|s|^2)$ and that completes the proof
    of the theorem.
\end{proof}

%%%%%%%%%%%%%%%%%%%%%%%%%%%%%%%%%%%%%%%%%%%%%%%%%%%%%%%%%%%%%%%%%%%%%%%%%%%%%%%%%%%%%%%%%%%%%

\section{The Operators $L$ and $M$}
\label{sec:operators}

In this appendix we prove Theorems \ref{thm:M-theorem} and
\ref{thm:L-theorem}. The proofs below are standard.

\begin{proof}[Proof of Theorem \ref{thm:M-theorem}]
    The fact that $M$ is positive follows immediately from its definition. We note that its being strictly positive is the result of
    restricting its domain to elements having mean zero.
\end{proof}

\begin{proof}[Proof of Theorem \ref{thm:L-theorem}]
    First, we note that  $L^n$ is clearly a positive self-adjoint operator. To see that it has discrete spectrum, we first note that the inclusion
    $H^2 \hookrightarrow L^2$ is compact for bounded domains in $\R^2$ with Lipschitz boundary (which certainly includes lattice cells). Then for any
    $z$ in the resolvent set of $L^n$, $(L^n - z)^{-1} : L^2 \to H^2$ is bounded and therefore $(L^n - z)^{-1} : L^2 \to L^2$ is compact.

    In fact we find the spectrum of  $L^n$ explicitly. We introduce the harmonic oscillator creation and annihilation operators %$L^n_\pm$ to be
    \begin{equation}
        L^n_\pm = \partial_{x_1} \mp i \partial_{x_2} \mp \frac{n}{2}x_1 + \frac{in}{2}x_2.
    \end{equation}
    One can verify that these operators satisfy the following.
    \begin{enumerate}
    \item $[L^n_+, L^n_-] = 2n$.
    \item $L^n - n = -L^n_+L^n_-$.
    \item $L^n - n = -L^n_-L^n_+$.
    \end{enumerate}
    As for the harmonic oscillator (see for example \cite{GS2}), this gives the explicit information about $\sigma(L)$ as stated in the theorem.

    For the dimension of the null space of $L$, we need the following lemma.
    \begin{lemma}
        $\Null (L^n - n) = \Null L^n_-$.
    \end{lemma}
    \begin{proof}
        If $L^n_-\psi = 0$, we immediately have $(L^n - n)\psi = -L^n_+L^n_-\psi = 0$. For the reverse inclusion we use the fact that
        $\| L^n_-\psi \|^2 = \langle L^n_-\psi, L^n_-\psi \rangle = \langle \psi, -L^n_+L^n_-\psi \rangle = \langle \psi, (L^n - n)\psi \rangle$.
    \end{proof}

    We can now prove the following.
    \begin{proposition} $ \Null L^n - n$ is given by %as on the r.h.s. of \eqref{nullL}
     \begin{equation} \label{nullL} \Null (L^n - n) =
     		\{\, e^{\frac{in}{2}x_2(x_1 + ix_2) }\sum_{k=-\infty}^{\infty} c_k e^{ki\sqrt{2\pi \Im \tau}(x_1 + ix_2)}\ |\ c_{k + n} = e^{in\pi\tau} e^{2ki\pi\tau} c_k \}.
        \end{equation}
     and therefore, in particular,
        $\dim_\C \Null L^n_- = n$.
    \end{proposition}
    \begin{proof}
        A simple calculation gives the following operator equation
        \begin{equation*}
            e^{\frac{n}{4}|x|^2}L^n_-e^{-\frac{n}{4}|x|^2} = \partial_{x_1} + i\partial_{x_2}.
        \end{equation*}
        This immediately proves that $\psi \in \Null L^n_-$ if and only if $\xi = e^{\frac{n}{4}|x|^2}\psi$ satisfies $\partial_{x_1}\xi + i\partial_{x_2}\xi = 0$.
        We now identify $x \in \R^2$ with $z = x^1 + ix^2 \in \C$ and see that this means that $\xi$ is analytic. We therefore define the entire function
        $\Theta$ to be
        \begin{equation*}
            \Theta(z) = e^{ -\frac{n (r^\tau)^2}{4\pi^2}z^2 } \xi \left( \frac{r^\tau z}{\pi} \right).
        \end{equation*}
         The quasiperiodicity of $\psi$ transfers to $\Theta$ as follows.
        \begin{subequations}
            \begin{equation*}
                \Theta(z + \pi) = \Theta(z),
            \end{equation*}
            \begin{equation*}
                \Theta(z + \pi\tau) =  e^{ -2inz } e^{ -in\pi\tau z } \Theta(z).
            \end{equation*}
        \end{subequations}

        To complete the proof, we now need to show that the space of the analytic functions which satsify these relations form a vector space
        of dimension $n$. It is easy to verify that the first relation ensures that $\Theta$ have a absolutely convergent Fourier expansion of the form
        \begin{align*}
            \Theta(z) = \sum_{k=-\infty}^{\infty} c_k e^{2kiz}.
        \end{align*}
        The second relation, on the other hand, leads to relation for the coefficients of the expansion. Namely, we have
        \begin{align*}
            c_{k + n} = e^{in\pi\tau} e^{2ki\pi\tau} c_k
        \end{align*}
        And that means such functions are determined solely by the values of $c_0,\ldots,c_{n-1}$ and therefore form an $n$-dimensional vector space.
    \end{proof}
    This completes the proof of Theorem \ref{thm:L-theorem}.
\end{proof}

%%%%%%%%%%%%%%%%%%%%%%%%%%%%%%%%%%%%%%%%%%%%%%%%%%%%%%%%%%%%%%%%%%%%%%%%%%%%%%%%%%%%%%%%%%%%%

\section{Fixing the Gauge}
\label{sec:alternate-proof}

We provide here an alternate proof of Proposition \ref{thm:fix-gauge}, largely based on ideas in \cite{Eil}. We begin by defining the function $B : \R \to \R$ to be
	\begin{align*}
		B(\zeta) = \frac{1}{r} \int_{0}^{r} \Curl A(\xi, \zeta)\,d\xi.
	\end{align*}
	It is clear that $b = \frac{1}{r\tau_2} \int_{0}^{r\tau_2} B(\zeta) \,d\zeta$. A calculation shows that $B(\zeta + r\tau_2) = B(\zeta)$.
%	\begin{align*}
%		B(x_2 + r\tau_2)
%				& = \frac{1}{r} \int_{0}^{r} \Curl A(\xi, x_2 + r\tau_2)\,d\xi \\
%				& = \frac{1}{r} \int_{-r\tau_1}^{r - r\tau_1} \Curl A(\xi + r\tau_1, x_2 + r\tau_2)\,d\xi \\
%				& = \frac{1}{r} \int_{-r\tau_1}^{r - r\tau_1} \Curl A(\xi, x_2)\,d\xi \\
%				& = \frac{1}{r} \int_{0}^{r} \Curl A(\xi, x_2)\,d\xi \\
%				& = B(x_2).
%	\end{align*}
	
%	\begin{align*}
%		b = \frac{1}{|\Omega|} \int_\Omega \Curl A(x_1,x_2)\,dx_1 dx_2 = \frac{1}{r\tau_2} \int_{0}^{r\tau_2} F(\zeta) \,d\zeta,
%	\end{align*}

	We now define $P = (P_1, P_2) : \R^2 \to \R^2$ to be
	\begin{align*}
	& P_1(x) = bx_2 - \int_{0}^{x_2} B(\zeta) \,d\zeta, \\
	& P_2(x) = \int_{\frac{\tau_1}{\tau_2} x_2}^{x_1} \Curl A(\xi, x_2)\,d\xi + \frac{\tau \wedge x}{\tau_2} B(x_2).
	\end{align*}
	A calculation shows that $P$ is doubly-periodic with respect to $\Lat$.

%	\begin{align*}
%		P_1(x_1 + r, x_2) = bx_2 - \int_{0}^{x_2} B(\zeta) \,d\zeta = P_1(x).
%	\end{align*}
%	\begin{align*}
%		P_1(x_1 + r\tau_1, x_2 + r\tau_2)
%				& = bx_2 + br\tau_2 - \int_{0}^{x_2 + r\tau_2} B(\zeta) \,d\zeta \\
%				& = P_1(x) + br\tau_2 - \int_{x_2}^{x_2 + r\tau_2} B(\zeta) \,d\zeta \\
%				& = P_1(x) + \int_{0}^{r\tau_2} B(\zeta) \,d\zeta - \int_{0}^{r\tau_2} B(\zeta) \,d\zeta \\
%				& = P_1(x).
%	\end{align*}
%	\begin{align*}
%		P_2(x_1 + r, x_2)
%				& = P_2(x) + \int_{x_1}^{x_1 + r} \Curl A(\xi, x_2)\,d\xi - r B(x_2) \\
%				& = P_2(x) + \int_{x_1}^{x_1 + r} \Curl A(\xi, x_2)\,d\xi - \int_{0}^{r} \Curl A(\xi, x_2)\,d\xi \\
%				& = P_2(x).
%	\end{align*}
%	\begin{align*}
%		P_2(x_1 + r\tau_1, x_2 + r\tau_2)
%				& = \int_{\frac{\tau_1}{\tau_2} x_2 + r\tau_1}^{x_1 + r\tau_1} \Curl A(\xi, x_2 + r\tau_2)\,d\xi + \frac{\tau \wedge x}{r\tau_2} \int_{0}^{r} \Curl A(\xi, x_2 + r\tau_2)\,d\xi \\
%				& = \int_{\frac{\tau_1}{\tau_2} x_2}^{x_1} \Curl A(\xi + r\tau_1, x_2 + r\tau_2)\,d\xi
%									+ \frac{\tau \wedge x}{r\tau_2} \int_{-r\tau_1}^{r-r\tau_1} \Curl A(\xi + r\tau_1, x_2 + r\tau_2)\,d\xi \\
%				& = \int_{\frac{\tau_1}{\tau_2} x_2}^{x_1} \Curl A(\xi, x_2)\,d\xi + \frac{\tau \wedge x}{r\tau_2} \int_{-r\tau_1}^{r-r\tau_1} \Curl A(\xi, x_2)\,d\xi \\
%				& = \int_{\frac{\tau_1}{\tau_2} x_2}^{x_1} \Curl A(\xi, x_2)\,d\xi + \frac{\tau \wedge x}{r\tau_2} \int_{0}^{r} \Curl A(\xi, x_2)\,d\xi \\
%				& = P_2(x).
%	\end{align*}
		
	We now define $\eta' : \R^2 \to \R$ to be
	\begin{align*}
		\eta'(x) = \frac{b}{2}x_1x_2 - \int_{0}^{x_1} A_1(\xi, 0) \,d\xi - \int_{0}^{x_2} A_2(x_1, \zeta) - P_2(x_1, \zeta) \,d\zeta.
	\end{align*}
	$\eta'$ satisfies
	\begin{align*}
		\nabla\eta = - A + A_0 + P.
	\end{align*}	
%	\begin{align*}
%		\frac{\partial\eta'}{\partial x_1}(x)
%				&= \frac{b}{2}x_2 - A_1(x_1, 0) - \int_{0}^{x_2} \frac{\partial A_2}{\partial x_1} (x_1, \zeta)
%						- \Curl A(x_1, \zeta) + B(\zeta) \,d\zeta \\
%				&= \frac{b}{2}x_2 - A_1(x_1, 0) - \int_{0}^{x_2} \frac{\partial A_1}{\partial x_2} (x_1, \zeta) + B(\zeta) \,d\zeta \\
%				&= \frac{b}{2}x_2 - A_1(x_1, x_2) - \int_{0}^{x_2} B(\zeta) \,d\zeta \\
%				&= - A_1(x_1, x_2) - \frac{b}{2}x_2 + P_1(x).
%	\end{align*}
%	\begin{align*}
%		\frac{\partial\eta}{\partial x_2}(x)
%				&= \frac{b}{2}x_1 - A_2(x_1, x_2) + P_2(x_1, x_2).
%	\end{align*}

	Now let $\eta''$ be a doubly-periodic solution of the equation $\Delta\eta'' = -\Div P$. Also let $C = (C_1, C_2)$ be given by
	\begin{equation*}
		C = - \frac{1}{|\Omega|} \int_\Omega P + \nabla\eta \,dx,
	\end{equation*}
	where $\Omega$ is any fundamental cell, and set $\eta''' = C_1x_1 + C_2x_2$.
	
	We claim that $\eta = \eta' + \eta'' + \eta'''$ is such that $A + \nabla\eta$ satisfies (i) - (iii) of the proposition. We first note that $A + \nabla\eta = A - A + A_0 + P + \nabla\eta'' + C$. By the above, $A' = P + \nabla\eta'' + C$ is periodic. We also calculate that $\Div A' = \Div P + \Delta\eta'' = 0$. Finally $\int A' = \int P + \nabla\eta - C = 0$.
	
	All that remains is to prove (iv). This will follow from a gauge transformation and translation of the state. We note that
	\begin{equation*}
		A_0(x + t) + A'(x + t) = A_0(x) + A'(x) + \frac{b}{2} \left( \begin{array}{c} -t_2 \\ t_1 \end{array} \right).
	\end{equation*}
	This means that $A_0(x + t) + A'(x + t) = A_0(x) + A'(x) + \nabla g_t(x)$, where $g_t(x) = \frac{b}{2} t \wedge x + C_t$ for some constant $C_t$. To establish (iv), we need to have it so that $C_t = 0$ for $t = r$, $r\tau$. First let $l$ be such that $r \wedge l = -\frac{C_r}{b}$ and $r\tau \wedge l = -\frac{C_{r\tau}}{b}$. This $l$ exists as it is the solution to the matrix equation
	\begin{equation*}
		\left( \begin{array}{cc} 0 & r \\ -r\tau_2 & r\tau_1 \end{array} \right)
			\left( \begin{array}{c} l_1 \\ l_2 \end{array} \right) =
			\left( \begin{array}{c} -\frac{C_r}{b} \\ -\frac{C_{r\tau}}{b} \end{array} \right),
	\end{equation*}
	and the determinant of the matrix is just $r^2\tau_2$, which is non-zero because $(r,0)$ and $r\tau$ form a basis of the lattice. Let $\zeta(x) = \frac{b}{2} l \wedge x$. A straight forward calculation then shows that $e^{i\zeta(x)}\psi(x + l)$ satisfies (iv) and that $A(x + l) + \nabla\zeta(x)$ still satisfies (i) through (iii). This proves the proposition.

%%%%%%%%%%%%%%%%%%%%%%%%%%%%%%%%%%%%%%%%%%%%%%%%%%%%%%%%%%%%%%%%%%%%%%%%%%%%%%%%%%%%%%%%%%%%%

%%%%%%%%%%%%%%%%%%%%%%%%%%%%%%%%%%%%%%%%%%%%%%%%%%%%%%%%%%%%%%%%%%%%%%%%%%%%%%%%%%%%%%%%%%%%%

\end{document}